\documentclass[10pt, conference, compsocconf]{IEEEtran}%

\usepackage{amsmath}
\usepackage{amsfonts}
\usepackage{amsthm}
\usepackage{multirow}
\usepackage{color}
\usepackage{graphicx}
\usepackage{subfig}
\usepackage[ruled,vlined,norelsize]{algorithm2e}
\usepackage{enumerate}
\usepackage{geometry}
\usepackage{tikz, pgfplots}
\usetikzlibrary{plotmarks}
\pgfplotsset{compat=newest}
\pgfplotsset{plot coordinates/math parser=false}
\newlength\figureheight
\newlength\figurewidth

\theoremstyle{definition}
\newtheorem {theorem} {Theorem}

\newtheorem{lemma}[theorem] {Lemma}

\newcommand{\mycell}[2][c]{%
  \begin{tabular}[#1]{@{}c@{}}#2\end{tabular}}

\begin{document}

%
%

\title{A Greedy Algorithm for Optimally Pipelining a Reduction}

\author{\IEEEauthorblockN{Bradley R. Lowery\IEEEauthorrefmark{1} and Julien Langou\IEEEauthorrefmark{1}}
\IEEEauthorblockA{
Mathematical and Statistical Sciences\\
University of Colorado Denver\\
Denver, CO, USA\\
\{Bradley.Lowery, Julien.Langou\}@ucdenver.edu  
}
\IEEEauthorblockA{\IEEEauthorrefmark{1}Research of this author was fully supported by the National Science Foundation \\
grant \# NSF CCF 1054864.}
}

\maketitle

%
%

\begin{abstract}
Collective communications are ubiquitous in parallel applications.  We present
two new algorithms for performing a reduction. The operation associated with our reduction needs
to be associative and commutative.  The two algorithms are developed
under two different communication models (unidirectional and bidirectional).
Both algorithms use a
greedy scheduling scheme.  For a unidirectional,
fully connected  network, we prove that our greedy algorithm is optimal when
some realistic assumptions are respected.  Previous algorithms fit the same
assumptions and are only appropriate for some given configurations. Our
algorithm is optimal for all configurations. We note that there are some
configuration where our greedy algorithm significantly outperform any existing
algorithms.
This result represents a
contribution to the state-of-the art.
For a bidirectional, fully connected network, we present a different greedy algorithm.
We verify by experimental simulations that our algorithm matches the time
complexity of an optimal broadcast (with addition of the computation). 
Beside reversing an optimal broadcast algorithm,
the
greedy algorithm is the first known reduction algorithm to experimentally attain this time
complexity.  Simulations show
that this greedy algorithm performs well in practice, outperforming any
state-of-the-art reduction algorithms. Positive experiments on 
a parallel distributed machine are also presented.
\end{abstract}


%
%

\section{Introduction}\label{intro}
Applications relying on parallel distributed communication require the use of
collective communications. In this paper, we focus on the reduction operation (MPI\_Reduce) 
in which each process holds a piece of data and these pieces of
data need to be combined using an associative operation to form the results on
the root process. We present
two new algorithms for performing a reduction. The operation associated with our reduction needs
to be associative and commutative.  The two algorithms are developed
under two different communication models (unidirectional and bidirectional).
Both algorithms use a
greedy scheduling scheme.

Collective communications like reduction can often be the bottleneck of
massively parallel application codes, therefore optimizing collective
communications can greatly improve the performance of such application codes.
The performance of a reduction algorithm is highly dependent on the machine
parameters (e.g., network topology) and the underlying architecture; this makes
the systematic optimization of a collective communication for a given machine a
real challenge. For a unidirectional,
fully connected  network, we prove that our greedy algorithm is optimal when
some realistic assumptions are respected.  Previous algorithms fit the same
assumptions and are only appropriate for some given configurations. Our
algorithm is optimal for all configurations. We note that there are some
configuration where our greedy algorithm significantly outperform any existing
algorithms. This result represents a
contribution to the state-of-the art.
For a bidirectional, fully connected network, we present a different greedy algorithm.
We verify by experimental simulations that our algorithm matches the time
complexity of an optimal broadcast (with addition of the computation). 
Beside reversing an optimal broadcast algorithm, the
greedy algorithm is the first known reduction algorithm to experimentally attain this time
complexity.  Simulations show
that this greedy algorithm performs well in practice, outperforming any
state-of-the-art reduction algorithms.
With respect to practical application, much work is still
needed in term of auto-tuning and configuring for various architectures (nodes
of multi-core, multi-port networks, etc.).

In the unidirectional context, we compared the greedy algorithm (uni-greedy) with three
standard algorithms (binomial, pipeline, and binary) using a linear model to
represent the point-to-point cost of communicating between processors.  Unlike
the standard algorithms, no closed-form expression exists for the completion
time of the greedy algorithm; therefore, simulations are used to compare the
algorithms. While we know that our algorithm is optimal in this context, these
simulations indicate the performance gain of the new algorithm over classic
ones.  In particular, for mid-size messages, the new algorithm is about 50\%
faster than pipeline, binary tree, or binomial tree.

In the bidirectional context, we again compare the greedy algorithm (bi-greedy) with the 
three standard algorithms as well as a reduce-scatter/gather (butterfly) algorithm.  For
the standard algorithms we see similar results as the unidirectional case.
The butterfly algorithm performs well for mid-size messages, but
has poor asymptotic behavior. 

In the experimental context, we have implemented the binomial, pipeline, uni-greedy,
and bi-greedy algorithms using OpenMPI version 1.4.3;  all algorithms are implemented
using point-to-point MPI functions: MPI\_Send and MPI\_Recv, or MPI\_Sendrecv.  Moreover, the
OpenMPI library provides a state-of-the-art implementation for a reduction.
We also, compare with an implementation of the butterfly algorithm.
Numerical comparisons of the greedy algorithm to OpenMPI's built
in function MPI\_Reduce as well as the binomial and pipeline algorithms show
the greedy algorithm is the best for medium size messages, confirming that the results
found theoretically applies in our experimental context.  However, the more simplistic algorithms
(binomial and pipeline) perform better for small and large messages, respectively.  
Finally, the implementation of the butterfly algorithm exhibits 
similar results as in the theoretical context.  

Finally, the idea of unequal segmentation is considered. Typically, when a message
is split into segments during a reduction operation, the segments are assumed
to be of equal size. We investigate if the greedy and pipeline algorithms can
be improved by allowing for segments to have unequal sizes. It turns out that,
for some parameter values, the greedy algorithm (in a unidirectional system) was optimized by unequal
segmentations.  This indicates that removing the equal segmentation assumption
from our theory can lead to better algorithms. However, the gains obtained were
marginal.  We also note that the pipeline algorithm was always optimized by
equal segmentation.

\section{Model}\label{model}
Unidirectional and bidirectional systems are two ways to describe how
processors are allowed to communicate with each other.  In a unidirectional
system, a processor can only send a message or receive a message at a given
time, but not both.  A bidirectional system assumes that at a given time a
processor may receive a message from a processor and send a message to another
(potentially different) processor simultaneously.  We assume an unidirectional
system for the initial sections and optimality proof (Sections~\ref{standard_algs} and~\ref{algor}).  In
Section~\ref{bidir} we adapt the algorithm to a bidirectional system.

The parallel system contains $p$ processors indexed from 0 to $p-1$. 
Additionally, we will assume the system is fully connected (every processor has a
direct connection to all other processors) and homogeneous. These
assumptions are not verified in practice on current architecture, however they
represent a standard theoretical framework and our experiments indicate that
they are valid enough in practice to develop useful algorithms. 

A linear model (Hockney~\cite{Hockney:94}) is used to represent the cost of a
point-to-point communication.  The time required for each communication is
given by $\alpha + \beta m$, where $\alpha$ is the latency (start up), $\beta$
is the inverse bandwidth (time to send one element of the message), and $m$ is the size
of the message. Since we are performing a reduction, computations are also
involved so we will also assume the time for computation follows a linear model
and is given by $\gamma m$, where $\gamma$ is the computation time for one
element of the message.  We will investigate the theoretical case when $\gamma
= 0$ (i.e., very fast processing units) as well as for nonzero $\gamma$.
Furthermore, we will assume that communication and computation can not overlap.

To achieve better performance the messages can be split into $q$ segments of size
$s_1,\dots,s_q$. In Section~\ref{algor} we show the greedy algorithm is optimal for any
segmentation. To obtain the best
performance of a given algorithm an optimal segmentation is determined. In Section~\ref{theor}
we restrict the optimization to equi-segmentation ($s_i = s, \; \forall \; i$), except the last
segment may possibly be smaller. 

\section{Collective Communications and Related Work}\label{cc}
A collective communication operation is an operation on data distributed between a set of processors (a collective).
Examples of collective communications are broadcast, reduce (all-to-one), reduce (all-to-all), scatter, and gather. \\
\textbf{Broadcast:} Data on one processor (the root) is distributed to the other processors in the collective.\\  
\textbf{Reduce:} Each processor in a collective has data that is combined entry-wise and the result is
stored on the root processor. \\
\textbf{All Reduce:} Same as reduce except that all processors in the collective contain a copy of
the result. \\
\textbf{Scatter:} Data that is on the root processor is split into pieces and the pieces are distributed between the
processors in the collective. The result is each processor (including the root processor) contains a piece of the
original data. \\
\textbf{Gather:} The reverse operation of a scatter. \\

Optimizing the reduction operation is closely related to optimization
of the broadcast operation.  Any broadcast algorithm can be reversed to 
perform a reduction.  
Bar-Noy et al.~\cite{BarNoy:2000} and Tr{\"a}ff and Ripke~\cite{TraffR:08} both provide 
algorithms that produce an optimal broadcast schedule for a bidirectional system. 
Here the messages are split into segments and the segments are broadcast in rounds. 
In both cases the optimality is in the sense 
that the algorithm meets the lower bound on the number of communication rounds.  
For the theoretical case when $\gamma = 0$, reversing an
optimal broadcast will provide an optimal reduction.  However, for $\gamma \ne 0$
this is no longer valid as an optimal schedule will most like take into account 
the computation.  

Rabenseifner~\cite{Rabenseifner:04} provides a reduce-scatter/gather algorithm (butterfly) which 
provides optimal load balancing to minimize the computation. 
Rabenseifner does not predefine a segmentation of the message, but
rather uses the techniques recursive-halving and recursive-doubling.  The algorithm 
is done in two phases.  In the first phase the message is repeatedly halved in size and
exchanged among processes.  At the end of the first phase the final result
is distributed among all the processors.  This phase is know as a reduce-scatter.  
The second phase gathers the results recursively doubling the size of the message.
In \cite{RabenseifnerTraff:04} Rabenseifner and Tr{\"a}ff
improve on the algorithm for non-power of two number of processors.

Sanders et al.~\cite{SandersST:pc:2009} introduce an algorithm to schedule to a reduction
(or broadcast) using two binary trees. The authors notice that two binary trees could be use
simultaneously and effectively reduce the bandwidth by a factor of two from a single binary tree.
The time complexity approaches the lower bound for large messages.  However, for small messages the
latency term is twice larger than optimal.  Also, the time complexity is never better than 
that of a reverse-optimal broadcast. 

Other models have been used to describe more complex machine architectures.
Heterogeneous networks, where processor characteristics are composed of 
different communication and computation rates, have been considered.  Beaumont et
al.~\cite{Beaumont:tpds:05,Beaumont:ipdps:05} consider
optimizing a broadcast and Legrand et al.~\cite{Legrand:2005} consider
optimizing scatter and reduce.  In these papers, the problem is formulated 
as a linear program and solved to maximize the throughput (number of 
messages broadcasting per time unit).
Also, higher dimensional systems have also been considered~\cite{ChanGGT:06}. 
Here a processor can communicate with more than 
one processor at a time.

The machine parameters and architecture vary between machines and one algorithm may perform
better on one machine versus another.  Accurately determining machine parameters can be a
difficult task.  In practice, auto-tuning for each machine is required to
obtain a well-performing algorithm. Vadhiyar et al.~\cite{Vadhiyar:2000} discuss how
to experimentally determine the optimal algorithm.  
Pjesivac-Grbovi{\'c} et al.~\cite{Pjesivac:IPDPS:05} compare various models
that can be used to auto-tune a machine. These models are more complex than the Hockney model
and can provide a more accurate model for the communication cost.  However, a linear model provides a
good bases for theoretically comparing algorithms. 
Other models, such as LogP, LogGP, and PLogP, each provided useful information
to help auto-tune a machine.  

\section{Standard Algorithms}\label{standard_algs}
Before introducing the greedy algorithms it is helpful to review
three standard algorithms (binomial, pipeline, and binary).

Since we are using a unidirectional system, we will call a processor either a sending
processor or a receiving processor. Sending and receiving processors are paired
together to perform a reduction. After a receiving processor receives a segment
it will combine its segment with the segment received. After a processor sends
a segment it is said to be reduced for that segment. For the reduction to be
completed each processor (except the root) must be reduced for each segment.
The three algorithms are described below.

\paragraph{Binomial} At a given time suppose there are $n$ processors
left to be reduced for a message, then $\lfloor n/2 \rfloor$ processors are
assigned as sending processors and $\lfloor n/2 \rfloor$ are assigned as
receiving processors. Afterwards there will be $\lceil n/2 \rceil$ processors
left for reduction and the process is repeat until $n=1$. Segmenting the
message would only increase the latency of the communication time since all
processors must be finished before the next segment can be started. Therefore,
we do not consider segmentation.

\paragraph{Pipeline} At a given time suppose that processors $k+1$ to
$p-1$ have been reduced for segment $s_i$. Processor $k$ is assigned as a
sending processor for segment $s_i$ and processor $k-1$ is assigned as a
receiving processor for segment $s_i$. 

\paragraph{Binary} At a given time suppose there are $2^n - 1$
processors left to be reduced for segment $s_i$, then $2^{n-1}$ processors are
assigned as sending processors and $2^{n-2}$ processors are assigned as
receiving processors. If $p \ne 2^n - 1$ for some $n$, then the initial step in
the algorithm reduces the ``extra'' processors so that in the next step there is
$2^n -1$, for some $n$, processors left for reduction. Since there are twice as
many sending processors as receiving processors, two processors will send to
the same receiving processor and the time for one step in the binary algorithm
is $2(\alpha + \beta s + \gamma s)$.

\subsection{Lower bounds and optimality of the segment size $s$.}

\begin{table*}[!t]
\centering
\resizebox{.75\textwidth}{!}{%
	\begin{tabular}{|l | l | c |}
	\hline
			{Binomial} & Time & $\lceil log_2 p \rceil (\alpha + \beta m + \gamma m)$ \\ \hline
\multirow{3}{*}{Pipeline} & Time & $(p-1)(\alpha + \beta s + \gamma s) + 2(q-1)(\alpha + \beta s + \gamma s)$ \\
\cline{2-3}
			& $s_{opt}$ & $\left(\dfrac{2m\alpha}{(p-3)(\beta+\gamma)}\right)^{1/2}$ \\ \cline{2-3}
			& $T_{opt}$ & $\left[ \left( (p-3)\alpha \right)^{1/2} + \left( 2m(\beta + \gamma) \right)^{1/2}\right]^2$ \\ \hline
			
\multirow{3}{*}{Binary} & Time & $2\left(\lceil log_2 (p+1) \rceil -1 \right)(\alpha + \beta s + \gamma s) + 4(q-1)(\alpha + \beta s +
\gamma s)$ \\ \cline{2-3}
			& $s_{opt}$ & $\left(\dfrac{2m\alpha}{(N-3)(\beta+\gamma)}\right)^{1/2}$ \\ \cline{2-3}
			& $T_{opt}$ & $2\left[ \left( (N-3)\alpha\right)^{1/2} + \left( 2m(\beta + \gamma)\right)^{1/2} \right]^2$ \\ \hline
        \end{tabular}}
\caption{Time analysis for standard algorithms in a unidirectional system, where $N = \lceil
    \log_2(p+1)\rceil$, $s_{opt}$ is the optimal equi-segment
size, and $T_{opt}$ is the time for the algorithm at $s_{opt}$. Formulae are valid for $p>3$.}\label{time}
\end{table*}

\begin{table*}[htbp]
\centering
\resizebox{.65\textwidth}{!}{%
\begin{tabular}{| c | c | c | c |}
\hline
  & Latency & Bandwidth & Computation \\
\hline
Reduce & $\lceil log_2 p \rceil \alpha$ & $2m\beta$ & $\frac{p-1}{p}m\gamma$\\
\hline
Binomial & $\lceil \log_2 p \rceil \alpha$ & $\lceil \log_2 p \rceil m \beta$ & $\lceil \log_2 p \rceil m \gamma$ \\
\hline
Pipeline & $(p-1)\alpha$ & $(p + 2m - 3)\beta$ & $(p+ 2m - 3)\gamma$\\
\hline
Binary & $ 2(N - 1) \alpha$ & $2(N + 2m - 3)\beta$ & $2(N + 2m - 3)\gamma$\\
\hline
\end{tabular}}
\caption{Lower bounds for each term in the time for a reduction as well as the lower bounds for each standard
algorithm in a unidirectioanl, where $N = \lceil \log_2(p+1)\rceil$.  Formulae are valid for $p>2$.}\label{lb}
\end{table*}
Given $p$, $\alpha$, $\beta$, $\gamma$, $m$, and $s$, the time for a reduction can be calculated for each algorithm. Using these
formulae the optimal segmentation size ($s_{opt}$) and optimal time ($T_{opt}$) can also be calculated.  When calculating the optimal
segmentation size we assume that the segments are of equal size.  A summary
of these values is given in Table~\ref{time}.  It is unknown what the optimal segmentation would be if the segments 
are of unequal size.  The formulae for time are upper bounds in the case when $m$ is not
divisible by $s$.  These formulae differ from the formulae that appear in ~\cite{TraffR:08,Pjesivac:IPDPS:05} by a factor of 2 in the bandwidth term
since the authors assume a bidirectional system. In a unidirectional system there is extra lag time.

To better understand the algorithms it is helpful to know the lower bounds for each term:$\alpha$, $\beta$, and
$\gamma$. The lower bounds for any reduction algorithm as well as other collective communications are discussed 
in~\cite{Chan:cluster:04}. The following rationales are used to obtain the lower bounds.

\paragraph{Latency} Every processor in a collective has at least one message that must be communicated. At any time step,
at most half of the remaining processors can send their messages to another processor to perform the reduction.
This gives the lower bound $\lceil \log_2 p \rceil \alpha$.

\paragraph{Bandwidth} Consider the case when $p=3$.  Each processor must be involved in a communication for 
each element of the message.  When two processors are involved in a communication the third processor is idle.  Therefore, 
there must be two communications performed for each element of the message, which cannot occur simultaneously.  This 
gives the lower bound $2m\beta$.  For $p > 3$, the lower bound holds since a message of the same size on more processors 
can not be communicated faster.  $p=2$ is a special case with a bandwidth term of $m\beta$.

\paragraph{Communication} Each element of the message must be involved in $p-1$ computations. Distributing the
computations evenly among all the processors gives the lower bound $\frac{p-1}{p}m \gamma$.

The lower bounds for the standard algorithms can be obtained using the following rationale. The latency term will be
minimized when no segmentation is used ($q=1$), since any segmentation will only increase the number of time steps and
hence increase the latency. The bandwidth and computation terms will be minimized when $s=1$, since this will maximize
the time when multiple processors are working simultaneously. Table \ref{lb} shows the lower bounds for any reduction
algorithm and the lower bounds for each of the standard algorithms.

Given machine parameters $\alpha$, $\beta$ and $\gamma$, one can select the algorithm that provides the best
performance. In practice, the latency is much larger than the bandwidth  $(\alpha \gg \beta)$. 
For small messages, $\alpha \gg \beta m$, the time for a reduction is
dominated by the latency. The binomial algorithm minimizes the latency term and will provide the best performance for
small messages. For large messages, $\alpha \ll \beta m$, the time for a reduction is dominated by the bandwidth. The
pipeline algorithm provides the smallest bandwidth term and will therefore give the best performance for large messages. For
medium size messages, the binary algorithm will provide the best performance. The binary algorithm has both
qualities that give binomial and pipeline good lower bounds for latency and bandwidth, respectively, only differing by a factor of two.

The values of $m$ that are considered ``small'' and ``large'' depend on the machine parameters
($\alpha$, $\beta$, $\gamma$, and $p$). 
Figure~\ref{fig:theoretical_best_standard_algorithm} shows which algorithm is best for $\alpha = 10$, $\beta = 1$, $\gamma =
0$ as a function of the number of processors $p$ and the message size $m$.  Here we are assuming that there is no time
required for the computation $(\gamma = 0)$.  

\begin{figure}[htbp]
	\centering

	\resizebox{\columnwidth}{!}{\includegraphics{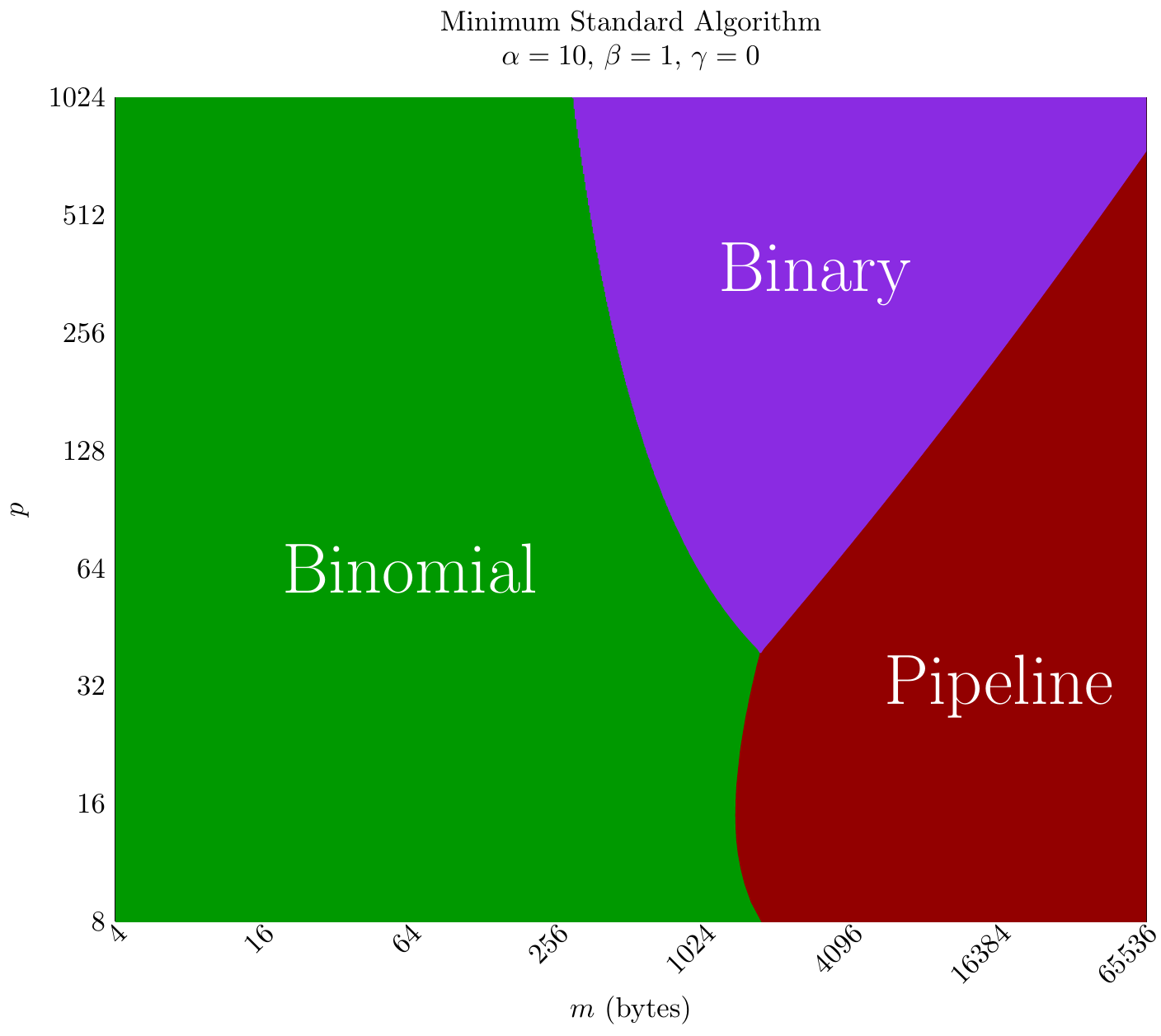}}
	\caption{\label{fig:theoretical_best_standard_algorithm}
	Regions where binomial or pipeline or binary is better in term of the number of
	processors ($p$) and the message size ($m$). For each 
	algorithm, 
	each $p$ and each $m$,
	the optimal segment size ($s_{opt}$) is obtained from the formulae given in Table \ref{time}.
	The machine parameters are $\alpha = 10$, $\beta = 1$, $\gamma = 0$ .
	}
\end{figure}

\section{A Greedy Algorithm for Unidirectional System}\label{algor}
When a message is split into segments a reduction algorithm consists of scheduling $(p-1)q$ 
send/receive pairs (one for each segment of the non-root processors). 
For a segment $i$ of size $s_i$, the sending
processor in a communication will be available after $t_{comm}(i) =\alpha + \beta
s_i$, and the receiving processor will be available after $t_{comm}(i) + t_{comp}(i)$, 
where $t_{comp} = \gamma s_i$.
We will add an extra requirement that segment $i$ must be reduced on a processor before that
processor can be involved in a reduction for segment $j$, $j > i$.
This requirement provides two benefits.  First, since the schedule that is generated for a given number of
processors and number of segments will vary, we must log the current state on each processor.
With this requirement we only need a small array with a length equal to the number of processors.  Second, 
if a processor has started but has not been reduced for multiple segments then each segments current result
must be stored on the processor.  This would require additional storage space on that processor.

Since we assume that each processor must reduce the segments in order, then a reduction can be scheduled considering 
one segment at a time.  Let $I$ be a set of $n$ processors that remain to send a given segment 
(or for the root, receive the final message). 
Let $x_i$ be the time that processor $i$ completed its previous task.  We define the state of these processors as
\[X^{(n)} = \{ x_{i} | i \in I \}. \] 
The state of the processors that have sent the segment is $S^{(n)}$.  
The permutation $P(X^{(n)})$ orders the elements of $X^{(n)}$ from 
smallest to largest.  The time complexity of an algorithm is given by the final state of the root processor.
A schedule for a send/receive pair is given by
\[ (X^{(n-1)}_k,S^{(n-1)}_k) = \textmd{reduce}(X^{(n)}_k,S^{(n)}_k,a,b,t). \]
Processor $a$ sends segment $k$ to processor $b$ starting at time $t$.
$X^{(n-1)}_k$ is the updated state of the processors that remain to send
segment $k$ and $S^{(n-1)}_k$ is the state of the processors that have sent segment $k$.  
We will assume that successive send/receive pairs are scheduled in the order that the communications are started.  
This way the state of a processor that is not a memeber of the send/receive pair, but whose current state
is less than $t$ is set equal to $t$.

\begin{lemma}\label{valid:sendrecv}
	$\textmd{reduce}(X^{(n)},S^{(n)},a,b,t)$, where $P(X^{(n)}) = (x_{g_1},x_{g_2},\dots,x_{g_n})$, is a 
	valid \newline send/receive pair if and only if 
	\begin{enumerate}[(i)]
		\item\label{v1} $x_a \le t$ and  $x_b \le t$;
		\item\label{v2} $a \ne 0$; and 
		\item\label{v3} $t \ge x_{g_2}$. 
	\end{enumerate} 
\end{lemma}

\begin{proof}
	\eqref{v1} requires that processor $a$ and $b$ have completed their previous task before $t$. \eqref{v2} states 
	that the root can never be a sending processor.  And \eqref{v1} implies \eqref{v3}.
\end{proof}

A greedy algorithm is obtained by scheduling every send/receive pair using $t= x_{g_2}$ 
in Lemma~\ref{valid:sendrecv}.  When a segment is done being scheduled then the final state of
the processors are used as the initial state for the next segment.  Algorithm~\ref{greedy-schedule} produces the 
unidirectional greedy (uni-greedy) schedule. A proof of optimality is given in Section~\ref{opt-proof}. 

\begin{algorithm}[!h]
	$X^{(p)}_1 = \{0,\dots,0\} \; S^{(p)}_1 = \emptyset $ \\
	\For{ $i = 1$ \KwTo $q$ }{
		\For{ $j = p$ {\bf downto} $2$ }{
			$P( X^{(j)}_i ) = ( x_{g_1}, x_{g_2}, \dots x_{g_j} )$ \\
			$ t = x_{g_2} $ \\
			$ ( X^{(j-1)}_i, S^{(j-1)}_i ) = \textmd{reduce}( X^{(j)}_i, S^{(j)}_i, g_1, g_2, t ) $ \\
		}
		$X^{(p)}_{i+1} = X^{(1)}_i \bigcup S^{(1)}_i$ \\
	}	
	\caption{Uni-greedy Schedule}\label{greedy-schedule}
\end{algorithm}

\begin{figure*}[!p]
    \centering
    \resizebox{.85\textwidth}{!}{
        \includegraphics{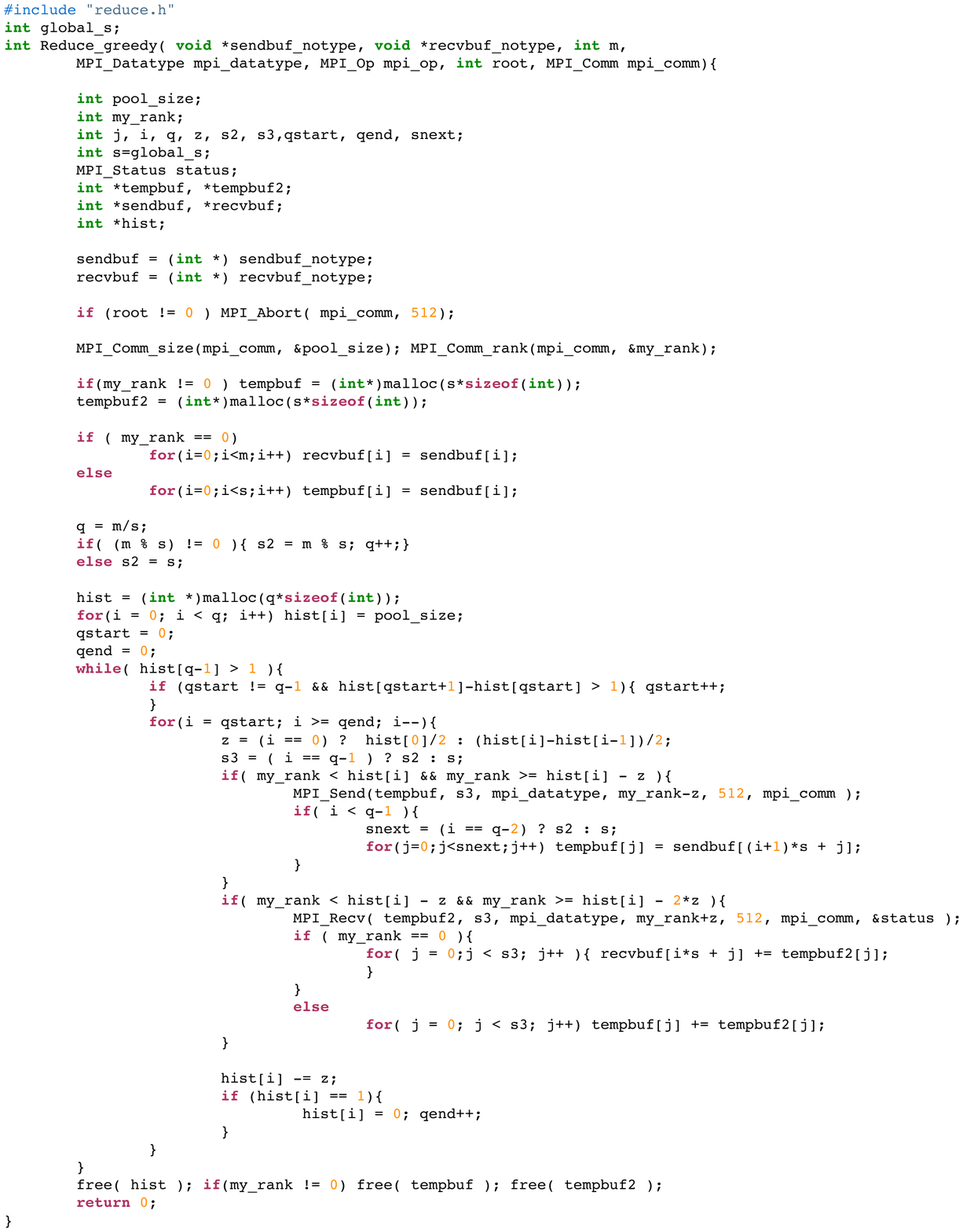}
    }
\caption{\label{fig:code_reduce_greedy}
C code for the uni-greedy reduction algorithm. 
The algorithm assumes that \texttt{MPI\_Op} is commutative (and associative as well, of course).
The algorithm allocates two extra buffers of size segment size.
The global variable \texttt{global\_s} is the size of a segment and needs to be initialized (if possible 
``tuned'') in advance. The implementation is restricted to
\texttt{root} being \texttt{0}, 
\texttt{MPI\_Datatype} being \texttt{int},
and
\texttt{MPI\_Op} being \texttt{+}.
These restrictions are not a consequence of the algorithm and can be removed.
}
\end{figure*}
\begin{figure*}[!t]
	\centering
    \resizebox{.8\textwidth}{!}{\includegraphics{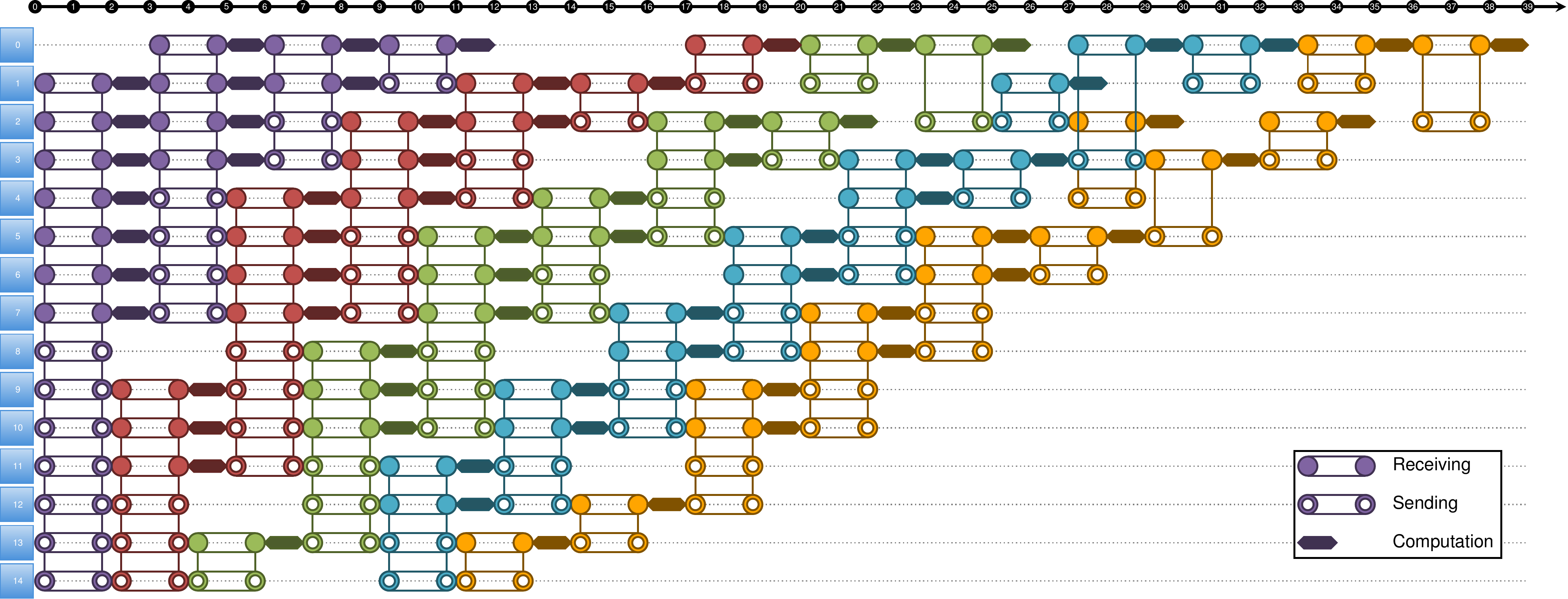}}
    \caption{\label{fig:GreedyTree} Uni-greedy algorithm for 15 processors and 5
equi-segments for $t_{comm} = 2$ and $t_{comp} = 1$. Each color represents a different segment. 
Filled in circles represent receiving processors, open circles represent sending processors, and
hexagons represent computation.}
\end{figure*}
We give the C code for the uni-greedy algorithm on the next page in Figure~\ref{fig:code_reduce_greedy}.
In this code, we are able to generate the greedy schedule dynamically during the reduction rather
than precomputing the schedule by Algorithm~\ref{greedy-schedule}.
The algorithm assumes that \texttt{MPI\_Op} is commutative (and associative as well, of
course).  The algorithm allocates two extra buffers of size segment size.  The variable \texttt{s} is the 
size of a segment and needs to be initialized (if possible ``tuned'') in advance.  The implementation is
restricted to \texttt{root} being \texttt{0}, \texttt{MPI\_Datatype} being
\texttt{int}, and \texttt{MPI\_Op} being \texttt{+}.  These restrictions are not
a consequence of the algorithm and can be removed.  In the next section we
analyze the theoretical time for the greedy algorithm for various values of
$\alpha$, $\beta$, and $\gamma$.

No closed-form expression for the time complexity of the uni-greedy algorithm exists 
like those given in Table~\ref{time} for the standard algorithms. 
Figure~\ref{fig:GreedyTree} gives an example for 15 processors and 5 equi-segments with $t_{comm} = 2$ and
$t_{comp} = 1$. From time 12 to 15 provides an example where if we allowed segments to be reduced out of order
then the uni-greedy algorithm would not be optimal.  Without this restriction processor
8 could send segment 4 to processor 0 and the overall time complexity could be reduced by 1.  
We choose to omit the figure that shows the improvement.

%
%

\subsection{Proof of Optimality}\label{opt-proof}

In this section, we prove that the uni-greedy algorithm is optimal.  The proof is
a variation of the proof in~\cite{CosnardRobert:86}.
Using the notation as in Lemma~\ref{valid:sendrecv}, we define 
\[ \textmd{optreduce}(X^{(n)}) = \textmd{reduce}(X^{(n)},S^{(n)},g_1,g_2,x_{g_2} ). \]

We also define the partial ordering on $\mathbb{R}^n$:
\[ X \le Y \iff x_i \le y_i, \; 1 \le i \le n. \]

\begin{lemma}\label{basiclemma}
	Let $X \in \mathbb{R}^n$ and $Y \in \mathbb{R}^n$ be nondecreasing  such that $X \le Y$ and let $a \in \mathbb{R}$ 
	and $b \in \mathbb{R}$ such that $a \le b$. If $x_k \le a \le x_{k+1}$ and $y_l \le b \le y_{l+1}$, then 
	\[(x_1,\dots,x_k,a,x_{k+1},\dots,x_n) \le (y_1,\dots,y_l,b,y_{l+1},\dots,y_n). \] 

\end{lemma}

\begin{proof}
	$\forall i \; \textmd{s.t.} \; i \le \min(k,l),$
	\[ x_i \le y_i. \]
	$\forall i \; \textmd{s.t.} \; i \ge \max(k,l)+1,$
	\[ x_i \le y_i. \]  
	
	If $k = l$, then $a \le b$ is the only other comparison needed. 
	
	If $k < l$, then 
	\[ a \le x_{k+1} \le y_{k+1}, \] 
	\[ x_l \le y_l \le b, \]  
	and $\forall i \; \textmd{s.t.} \; k < i < l,$
	\[ x_i \le y_i \le y_{i+1}. \]

	If $k > l$, then 
	\[ a \ge x_k \ge x_{l+1} \ge y_{l+1} \ge b. \]  
	By assumption $a \le b$, therefore $a=b$. Hence, 
	\[ x_{l+1} = a = b, \] 
	\[ a = x_k \le y_k, \] 
	and $\forall i \; \textmd{s.t.} \; l+1 < i \le k,$ 
	\[ x_i = a = b \le y_{i-1}. \]
	
	Hence, for all cases the inequality is satisfied.	
\end{proof}

\begin{lemma}\label{sendrecv}
	Let 
	\[ (X^{(n-1)},S^{(n-1}) = \textmd{optreduce}(X^{(n)}) \] 
	and 
	\[ (Y^{(n-1)},R^{(n-1)}) = \textmd{reduce}(Y^{(n)},R^{(n)},c,d,t_n). \] 
	
	\noindent If $P(X^{(n)}) \le P(Y^{(n)})$, then
	\begin{enumerate}[(i)]
		\item\label{remaining}  $P(X^{(n-1)}) \le P(Y^{(n-1)})$; and 
		\item\label{senders} 	$P(S^{(n-1)}) \le P(R^{(n-1)})$.
	\end{enumerate}
\end{lemma}

\begin{proof}
	Let $P(X^{(n)}) = (x_{g_1},x_{g_2},\dots,x_{g_n})$ and $P(Y^{(n)}) = (y_{h_1},y_{h_2},\dots,y_{h_n})$. 
	The receiving processors will be updated to $x_{g_2}+t_{comp}+t_{comm}$ and $t+t_{comp}+t_{comm}$,
	respectively.
	Let 
	\[ X^* = P( X^{(n-1)} \setminus \{ x_{g_2}+t_{comp}+t_{comm}\} ) \]  
	and 
	\[ Y^* = P( Y^{(n-1)} \setminus \{ t+t_{comp}+t_{comm}\} ). \]  
	
	We will show that $X^* \le Y^*$.
	Let $x_{g_k}$ be the smallest state greater than $x_{g_2}$ and $y_{h_l}$ be the smallest state 
	greater than $t$. Then 
	\[ X^* = (x_{g_2}, \dots, x_{g_2}, x_{g_k},\dots, x_{g_n} ) \]
	\noindent and
	\[ Y^* = (t,\dots, t, y_{h_l}, \dots, y_{h_n}). \]
	$\forall i \; \textmd{s.t.} \; i <   \min(k,l),$
	\[ x_{g_2} \le y_{h_2} \le t. \]
	$\forall i \; \textmd{s.t.} \; i \ge \max(k,l),$
	\[ x_{g_i} \le y_{h_i}. \]
	If $k < l$, then $\forall i \; \textmd{s.t.} \; k \le i < l,$ 
		\[ x_{g_i} \le y_{h_i} \le t. \]
	If $k > l$, then $\forall i \; \textmd{s.t.} \; l \le i < k,$ 
		\[ x_{g_i} = x_{g_2} \le y_{h_2} \le t \le  y_{h_i}. \]
	We have shown $X^* \le Y^*$ and $x_{g_2}+t_{comp}+t_{comm} \le t+t_{comp}+t_{comm}$.
	Therefore, by Lemma~\ref{basiclemma}, $P(X^{(n-1)}) \le P(Y^{(n-1)})$.

	Now we will prove the second claim by induction. The state of the sending processors are updated to 
	$x_{g_2} + t_{comm}$ and $t + t_{comm}$ and clearly, $x_{g_2} + t_{comm} \le t + t_{comm}$.
	For $n=p$, we have $S^{(p-1)} \le R^{(p-1)}$.  Assume that 
	$P(S^{(n)}) \le P(R^{(n)})$, by Lemma~\ref{basiclemma}, \eqref{senders} follows.

\end{proof} 

We can now define an iteration on segment $k$ with initial state $X_k^{(p)}$ as repeatedly applying $p-1$ send/receive
pairs.  Let $X_{k+1}^{(p)} = \textmd{optiter}(X_k^{(p)})$ be the iteration that applies $\textmd{optreduce}$ 
at every step. The output is the collection of states of the sending processor from each step and the state of the
root processor after finishing the final computation.  Let
$X_{k+1}^{(p)} = \textmd{iter}(X_k^{(p)})$ be an iteration that applies any set of $\textmd{reduce}$ operations. 

\begin{lemma}\label{iter}
	If $P(X_k^{(p)}) \le P(Y_k^{(p)})$ and $X_{k+1}^{(p)} = \textmd{optiter}(X_k^{(p)})$ and 
	$Y_{k+1}^{(p)} = \textmd{iter}(Y_k^{(p)})$, then 
	$P(X_{k+1}^{(p)}) \le P(Y_{k+1}^{(p)})$.
\end{lemma}

\begin{proof}
	By Lemma~\ref{sendrecv} the state vectors of the remaining processors at every step will satisfy the 
	conditions of Lemma~\ref{sendrecv}.  Hence, we can conclude that $S^{(1)}_k \le R^{(1)}_k$ and 
	$X^{(1)}_k \le Y^{(1)}_k$, where 
	$( X^{(1)}_k, S^{(1)}_k ) = \textmd{optreduce}(X^{(2)}_k)$ and
	$( Y^{(1)}_k, R^{(1)}_k ) = \textmd{reduce}(Y^{(2)}_k)$.
	The initial states for segment $k+1$ are $X_{k+1}^{(p)} = S^{(1)}_k \bigcup X^{(1)}_k$ and
	$Y_{k+1}^{(p)} = R^{(1)}_k \bigcup Y^{(1)}_k$.
	By Lemma~\ref{basiclemma}, \\ 
	$P(X_{k+1}^{(p)}) \le P(Y_{k+1}^{(p)}).$
\end{proof}

A reduction algorithm is obtained by repeatedly applying an iteration for every segment using the ending state vector
of one segment as the initial state vector of the following.  We obtain the greedy algorithm by always applying
$\textmd{optiter}$.  Any reduction algorithm is obtained by repeatedly applying any iteration. 
 
\begin{theorem}\label{optgreedy}
	In a unidirectional, fully connected, homogeneous system the time complexity of the uni-greedy algorithm minimal among all reduction algorithms
	that reduce segments in order.
\end{theorem} 

\begin{proof}
	$X_k^{(p)}$ is the state of the processors at the start of segment $k$ when applying the greedy algorithm 
	and $Y_k{(p)}$ is the state of the processors at the start of segment $k$ when applying another reduction algorithm.  
	We assume that the initial state of the processors before any reduction is done is zero. 
	That is, $X^{(p)}_1 = Y^{(p)}_1 = \{ 0,\dots 0 \}$. 
	By Lemma~\ref{iter}, $P( X_2^{(p)} ) \le P( Y_2^{(p)} )$ and by induction
	$X_{q+1}^{(p)} \le Y_{q+1}^{(p)}$. Hence, $x_0 = \max(X_{q+1}^{(p)}) \le \max(Y_{q+1}^{(p)}) = y_0$.
\end{proof}


\subsection{Theoretical Results}\label{theor}
In Section~\ref{opt-proof} we showed that given any segmentation of a message the uni-greedy algorithm is optimal
among algorithms that reduce the segments in order on a processor.  However, selecting the optimal
segmentation over all possible segmentations is difficult.  To avoid this difficulty, we restict ourselves to considering
only equi-segmentation in this section.  In Section~\ref{uneq} we will investigate unequal segmentation.

We compare the optimal segmented uni-greedy algorithm with optimally segmented
binomial, binary, and pipeline algorithms. For pipeline and binary,
the optimal segment size are obtained by the formulae in Table \ref{time}. 
For binomial, the messages where always reduced with a single segment.
The parameters used for the theoretical experiments were: $\alpha = 0,1,10,100,1000$, $\beta = 1$, $\gamma = 0,1$, $p
= 2^k, \mbox{s.t.} \; k = 2,3,\dots,10$, and $m = 2^k \; \mbox{bytes}, \mbox{s.t.} \; k = 2,3,\dots,16$. 

\begin{figure*}[!tb]
    \centerline{
        \subfloat[Message size versus ratio]{
            \label{fig:theo_p64}
			\resizebox{.45\textwidth}{!}{\includegraphics{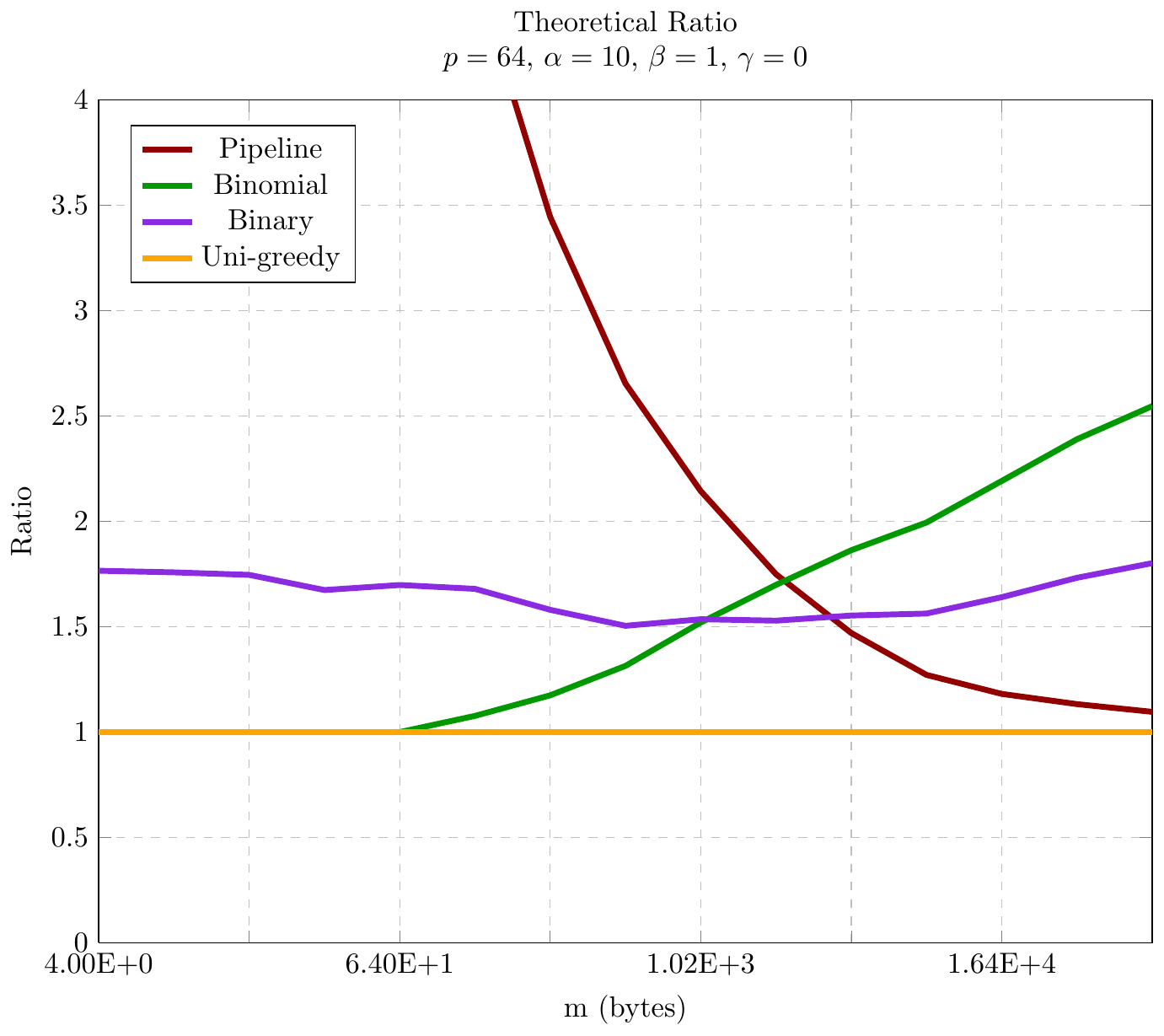}}
        }
        \hfil
        \subfloat[Message size versus time]{
            \label{fig:theo_p64_time}
			\resizebox{.45\textwidth}{!}{\includegraphics{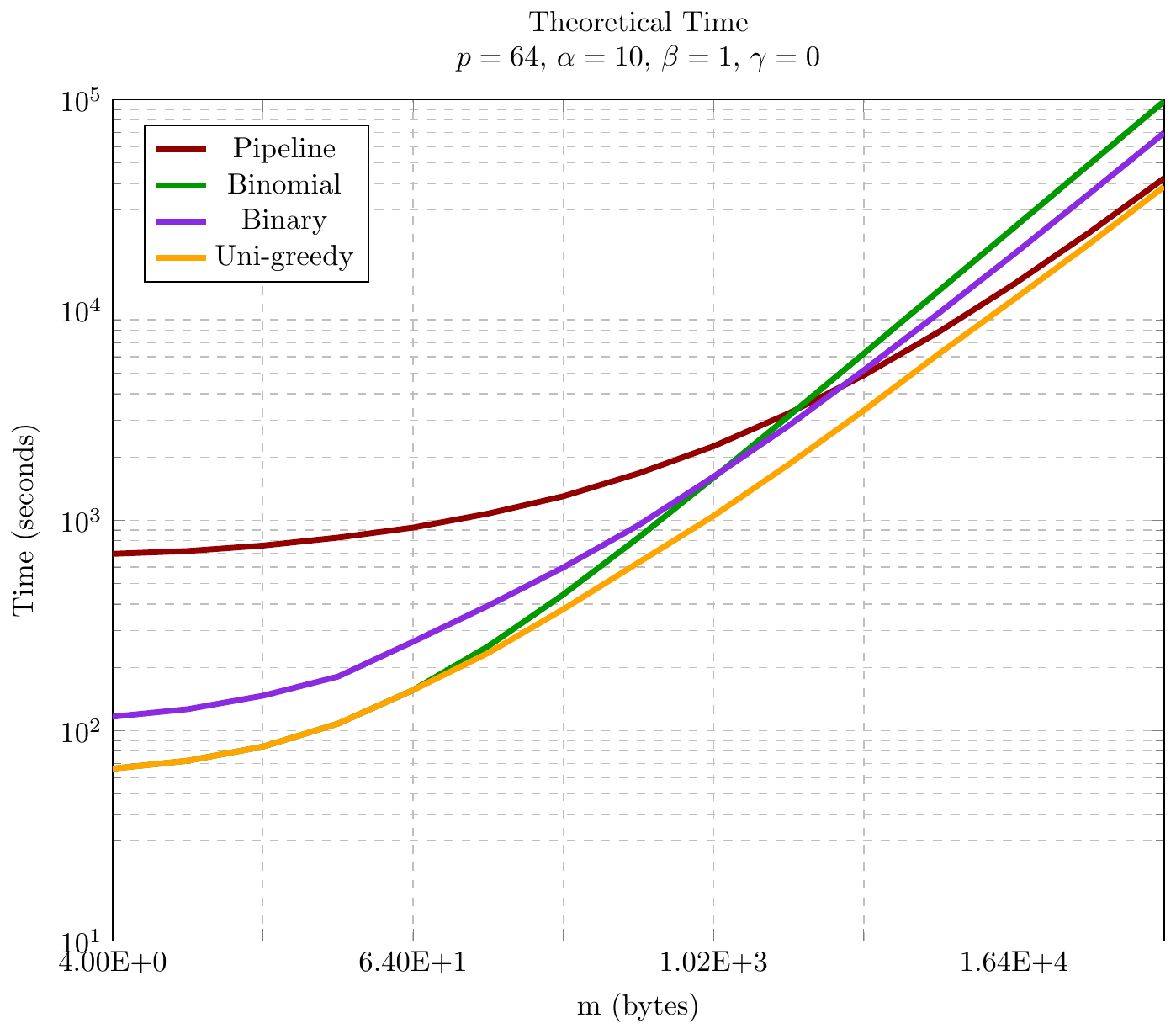}}
        }
    }
    \caption{Theoretical results for $p = 64$ plotting message size ($m$) versus ratio~\eqref{fig:theo_p64} and time~\eqref{fig:theo_p64_time}
	in a unidirectional system. 
    The machine parameters are fixed at $\alpha = 10$, $\beta = 1$, and $\gamma = 0$.}
    \label{fig:theoretical_p64}
\end{figure*}
Figures~\ref{fig:theo_p64}, \ref{fig:theo_p64_time}, and \ref{fig:theo_3d} show the results for when $\alpha = 10$,
$\beta = 1$, and $\gamma = 0$. In Figure~\ref{fig:theo_p64} and \ref{fig:theo_p64_time}, the number of processors is
fixed at 64. In Figure~\ref{fig:theo_p64} the algorithms are compared by plotting the ratio between the algorithms time and
the time for the uni-greedy algorithm. For small messages the optimal segment size for uni-greedy is 1, and uni-greedy is
exactly the binomial algorithm. So for small messages, ($m < 64$ bytes in this
example) uni-greedy and binomial have the same time. For large messages ($m >
1.64\times10^4$ bytes in this example), the poor start up of the pipeline
algorithm is negligible and the pipeline algorithm approaches the uni-greedy
algorithm asymptotically.  The middle region is where uni-greedy is the most
relevant, seeing an increase in performance over the standard algorithms up to
approximately 50\% faster.  Figure~\ref{fig:theo_p64_time} plots the time for
each algorithm.

Figure~\ref{fig:theo_3d} gives the results for each value of $p$ and $m$, the uni-greedy algorithm is
compared with the best of the standard algorithms by plotting

\[ \textmd{Ratio} = \dfrac{ \textmd{Minimum time of the standard algorithms} }{ \textmd{Time for the Greedy Algorithm} }. \]

A ratio of 1 (white in the figure) indicates that a standard algorithm has the same time as the uni-greedy algorithm. Larger numbers indicate a more significant improvement.   Again the region that uni-greedy provides the greatest improvement is for ``medium'' size messages. As the number of processors
increase, the range of message sizes where uni-greedy provides the most improvement (ratio $\approx$ 1.5) increases. This is
largely due to the difference in the number of time steps required to reduce the first segment of uni-greedy ($\lceil log_2 p
\rceil$) versus that of pipeline ($p-1$). For larger $p$, pipeline requires a larger message before it is able to ``make up'' for the extra time to reduce the first segment.

\begin{figure}[htbp]
	\centering
			\resizebox{.5\textwidth}{!}{\includegraphics{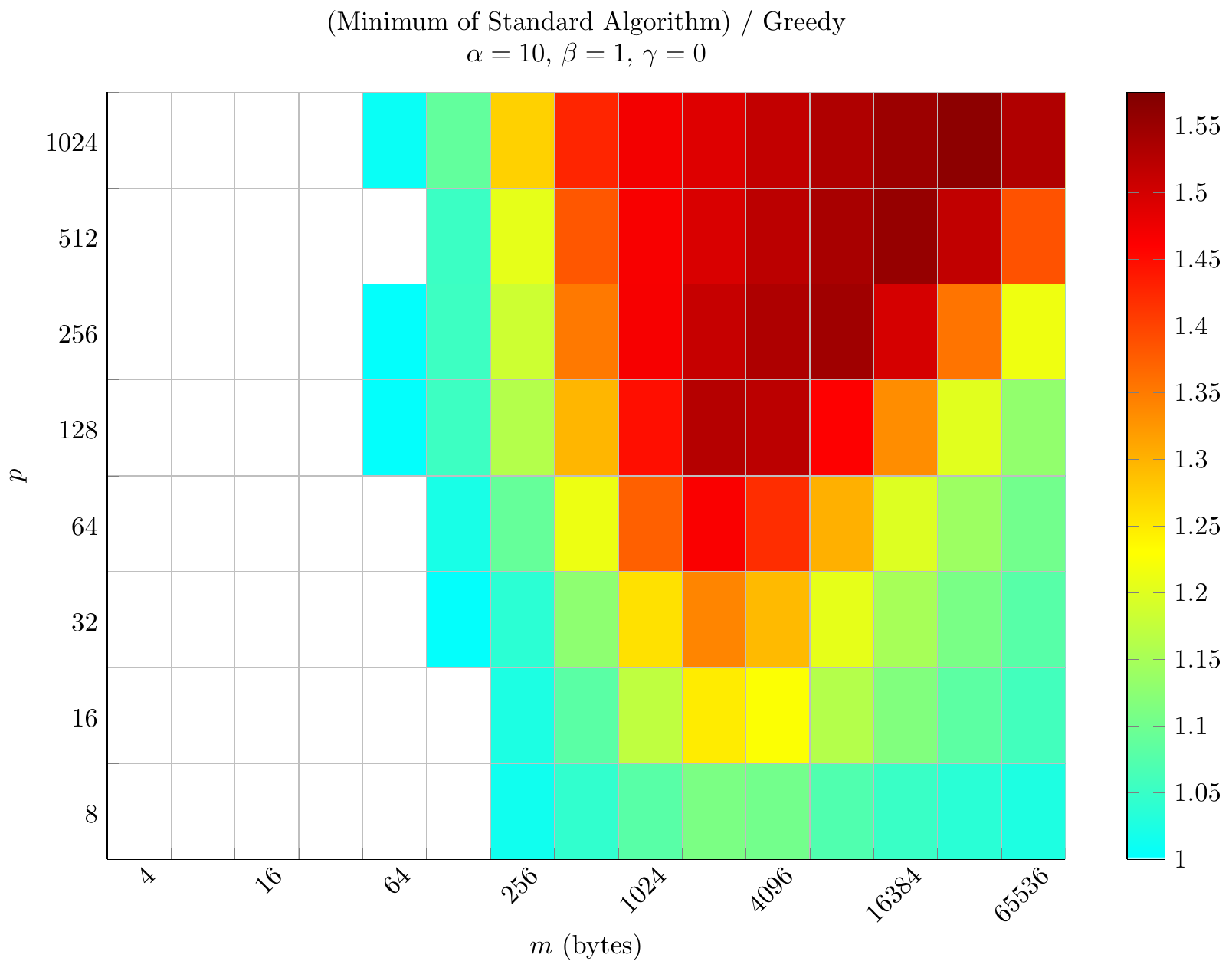}}
			\caption{Ratio between the time for the best standard algorithm and the uni-greedy algorithm for 
			varying values of $p$ (number of processors) and $m$ (message size). The machine parameters 
			are fixed at $\alpha = 10$, $\beta = 1$, and $\gamma = 0$.}
	\label{fig:theo_3d}
\end{figure}


\section{Bidirectional System}\label{bidir}
\begin{table*}[!bt]
\centering
\resizebox{.6\textwidth}{!}{%
	\begin{tabular}{|l | l | c |}
	\hline
	{Binomial} 
		& Time & 
			$\lceil log_2 p \rceil (\alpha + \beta m + \gamma m)$ \\ \hline
	\multirow{3}{*}{Pipeline} 
		& Time & 
			$(p+q-2)(\alpha + \beta s + \gamma s)$ \\\cline{2-3}
		& $s_{opt}$ & 
			$\left(\dfrac{m\alpha}{(p-2)(\beta+\gamma)}\right)^{1/2}$ \\ \cline{2-3}
		& $T_{opt}$ & 
			$\left[ \left( (p-2)\alpha \right)^{1/2} + \left( m(\beta + \gamma) \right)^{1/2}\right]^2$ \\ \hline
	\multirow{3}{*}{Binary} 
		& Time & 
			$2 \left(\lceil log_2 (p+1) \rceil +q - 1 \right)(\alpha + \beta s + \gamma s)$ \\ \cline{2-3}
		& $s_{opt}$ & 
			$\left(\dfrac{m\alpha}{(N-2)(\beta+\gamma)}\right)^{1/2}$ \\ \cline{2-3}
		& $T_{opt}$ & 
			$2\left[ \left( (N-2)\alpha\right)^{1/2} + \left( m(\beta + \gamma)\right)^{1/2} \right]^2$ \\ \hline
	\multirow{3}{*}{\mycell[c]{Bi-greedy \& \\Optimal Broadcast} } 
		& Time & 
			$(\lceil\log_2 p \rceil + q-1)(\alpha + \beta s + \gamma s)$ \\ \cline{2-3}
		& $s_{opt}$ & 
			$\left(\dfrac{m\alpha}{(\lceil\log_2 p \rceil- 1)(\beta+\gamma)}\right)^{1/2}$ \\ \cline{2-3}
		& $T_{opt}$ & 
			$\left[ \left( (\lceil\log_2 p \rceil - 1)\alpha \right)^{1/2} 
				+ \left( m(\beta + \gamma) \right)^{1/2}\right]^2$ \\ \hline
	{Butterfly} 
		& Time & 
			$2\lceil\log_2 p\rceil \alpha + 2\frac{p-1}{p} \beta m + \frac{p-1}{p} \gamma m$ \\ \hline
	\end{tabular}}
\caption{Time analysis for bidirectional reduction algorithms, where $N = \lceil \log_2(p+1)\rceil$, $s_{opt}$ is the optimal equi-segment
size, and $T_{opt}$ is the time for the algorithm at $s_{opt}$. Formulae are valid for $p>3$. The reduce-scatter/gather formula is only a lower bound when 
$p$ is not a power of two.}\label{tab.bi.time}
\end{table*}

In this section, we can adapt the uni-greedy algorithm presented in Section~\ref{algor} 
under the unidirectional context to an algorithm more suited for a bidirectional system.
We wish to compare the new algorithm to the current state-of-the-art.  We again have 
the standard algorithms: binomial, binary, and pipeline. We also consider  
a butterfly algorithm~\cite{Rabenseifner:04,RabenseifnerTraff:04} which minimizes the 
computation term of the reduction.
All of the algorithms mentioned so far work for non-commutative operations.  If the operation is
commutative, then any broadcast algorithm can be used (in reverse) for a reduction. 
Bar-Noy, et al.~\cite{BarNoy:2000} and Tr{\"a}ff and Ripke~\cite{TraffR:08} provide broadcast algorithms
that match the lower bound on the number of communication rounds for broadcasting $q$ segments among $p$ processors.
The time complexity for such a reduction algorithm is
\[ ( \lceil \log_2 p \rceil + q - 1 )(\alpha + \beta m + \gamma m). \]
From theoretical simulations it is seen that our new algorithm (bi-greedy) has the
same time complexity of a reverse optimal broadcast.
The rest of this section is organized as follows: In Section~\ref{bi-theo} we analyze the time complexity of 
the reduction algorithms. In Section~\ref{bi-greedy} we introduce the new algorithm (bi-greedy).

\subsection{Theoretical Results}\label{bi-theo}

We begin by summarizing the time complexity for each algorithm in Table~\ref{tab.bi.time}.
Binomial and butterfly do not have a choice on the segment size, so the optimal segment/time rows are
not included for those algorithms.
Binomial, binary, pipeline, and bi-greedy, all work in the same way: the message is split into $q$ segments of
equi-segment size and segments are communicated in rounds, where the time for one round is $\alpha + \beta s + \gamma s$. 
For these algorithms it is clear that bi-greedy is the best since it require the minimum number of rounds for any
number of segments.  
The butterfly algorithm does not start with a fixed segmentation of the message, but rather recursively 
halves the message size and exchanges the message between processors.  This method allows the computation
to be distributed evenly among the processors and hence minimizing the computational term~\ref{lb}.  
If the computational rate is small (i.e. $\gamma$ is large) than minimizing the computation term
can be advantageous.  

Figure~\ref{fig:theo-bi-ratio} and 
\ref{fig:theo-bi-time} show the results for when $p = 64$, $\alpha = 50000$, $\beta = 6$ and $\gamma = 1$. 
In Figure~\ref{fig:theo-bi-ratio} the ratio between the time of the algorithm and the time of bi-greedy is shown.
When comparing bi-greedy with binomial, binary, and pipeline, we see similar results as in the unidirectional 
case.  For small messages ($m < 10$ KB in this example), bi-greedy and binomial have the same time.
For large messages ($m > 10^5$ KB in this example), pipeline approaches the bi-greedy algorithm asymptotically.
For medium length messages ($200 < m < 4000$ KB in this example), bi-greedy provides approximately 
a 50\% increase in performance.
In this example, the butterfly algorithm performs well for medium size messages ($m \approx 300$~KB in this example) 
but does not provide good asymptotic behavior for large and small message sizes.  This is not a surprise as
Rabenseifner provides tuning experiments to switch to the binomial algorithm for small messages, and the 
pipeline algorithm for large messages~\cite{Rabenseifner:04}.  

\begin{figure*}[htbp]
    \centerline{
        \subfloat[Message size versus ratio]{
            \label{fig:theo-bi-ratio}
			\resizebox{.45\textwidth}{!}{\includegraphics{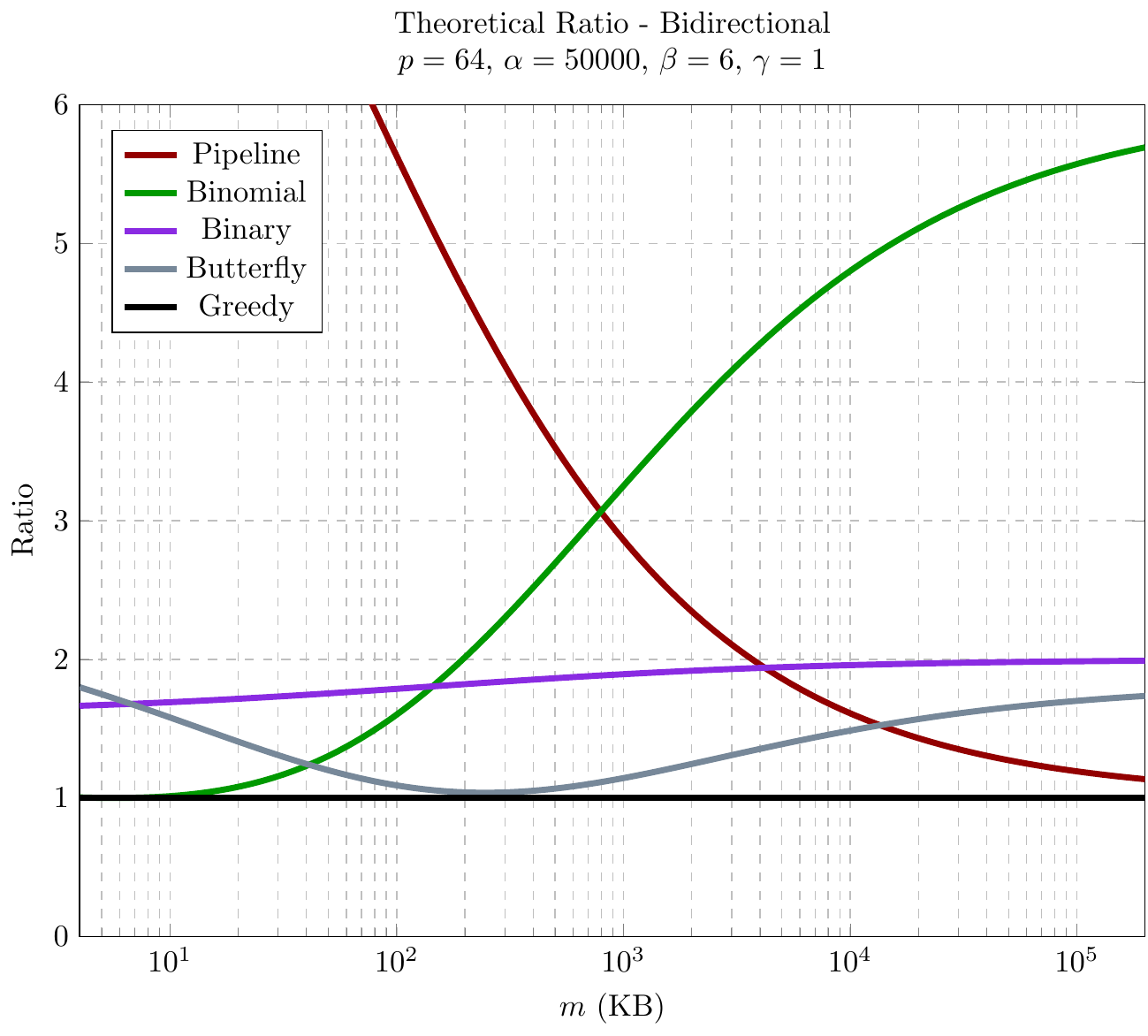}}
        }
        \hfil
        \subfloat[Message size versus time]{
            \label{fig:theo-bi-time}
			\resizebox{.45\textwidth}{!}{\includegraphics{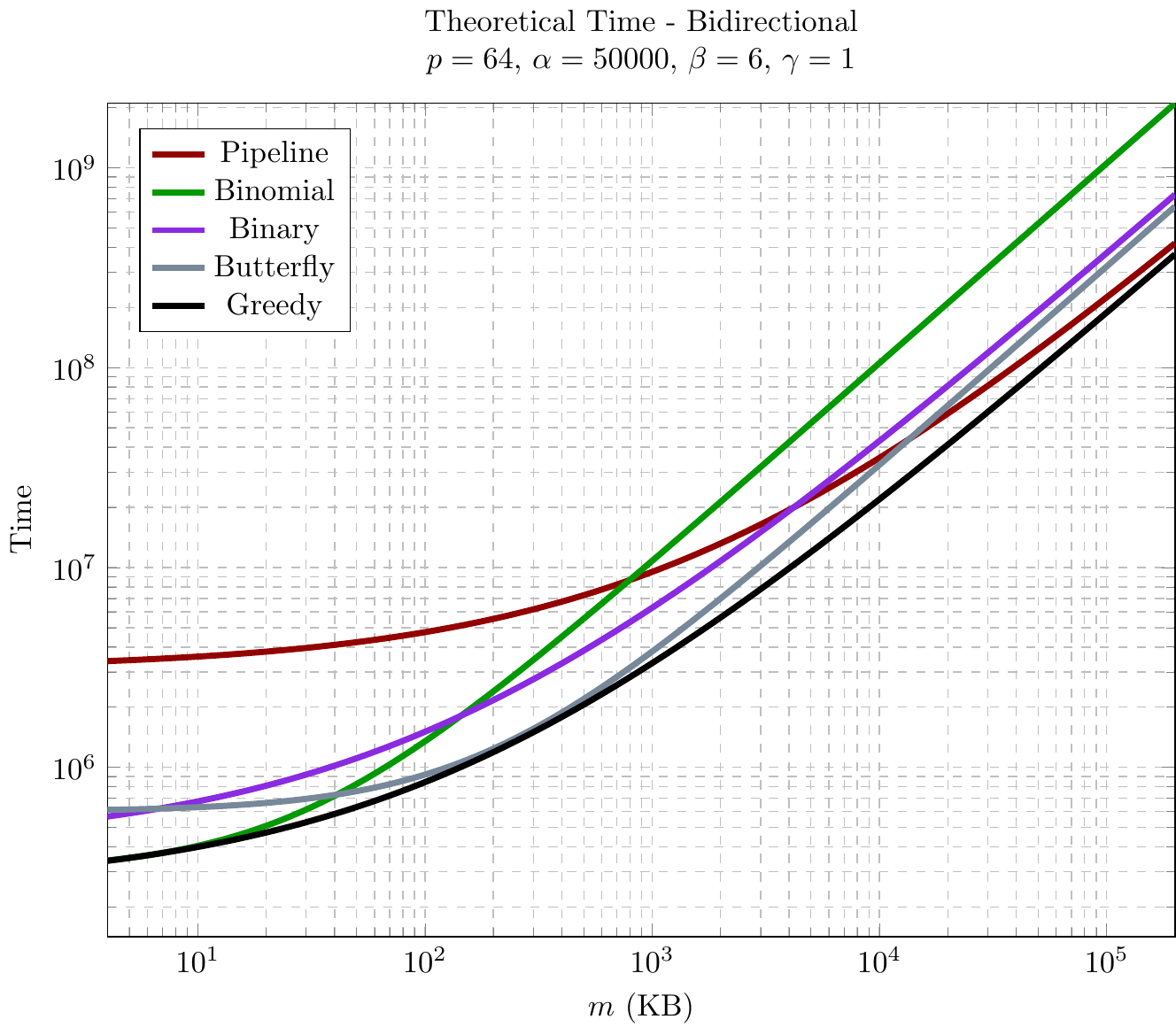}}
        }
    }
    \caption{Theoretical results for bidirectional system when $p = 64$ plotting message size ($m$) versus 
	ratio~\eqref{fig:theo-bi-ratio} and time~\eqref{fig:theo-bi-time}. 
    The formula for the reference line is $\lceil \log_2 p \rceil \alpha + m\beta + \frac{p-1}{p}m\gamma,$ which
	is the sum of the lower bounds for each term: $\alpha$, $\beta$, $\gamma$.
    The machine parameters are fixed at $\alpha = 50000$, $\beta = 6$, and $\gamma = 1$.}

	\label{fig:theoretical_p64}
\end{figure*}

For some machine parameters the butterfly algorithm can perform better than bi-greedy for some message sizes.
If the ratio, $\beta / \gamma$, is small enough then for some message sizes the butterfly algorithm
will be better than bi-greedy. 
To be more precise as to what constitutes a small ratio we consider Figure~\ref{fig.compare.raben}. 
Above the blue line bi-greedy will always be better 
than the butterfly algorithm.  Below the blue line, there exists a message size, $m$, 
such that the butterfly algorithm is better than bi-greedy 
(in Figure~\ref{fig:theoretical_p64} butterfly would dip below bi-greedy).  
For up to 1000 processors the ratio would have to be less than 5 for the butterfly algorithm to be better.  
In practice, a ratio larger than 10 is 
more likely, even for single core processors. 

\begin{figure}[htbp]
	\centerline{
	\resizebox{\columnwidth}{!}{\includegraphics{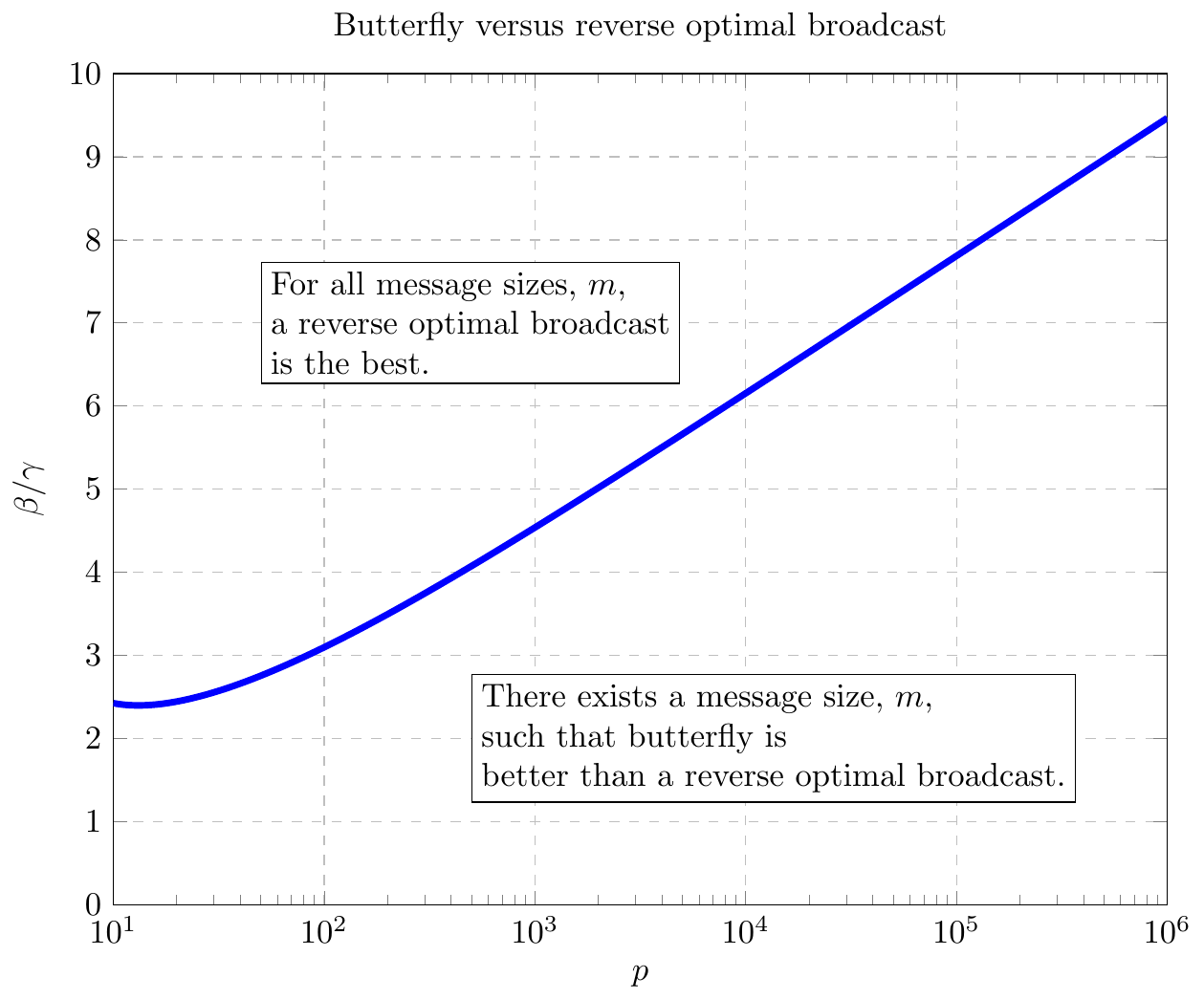}}
	}
	\caption{Parameters for which there exists a message size, $m$, such that the butterfly algorithm
		is better than a reverse optimal broadcast algorithm.}
	\label{fig.compare.raben}
\end{figure}

%
%
%
%

\subsection{The Bi-greedy Algorithm}\label{bi-greedy}

In the unidirectional case, the optimal algorithm attempted to fill all available ports 
at any time. This is the goal when adapting to the bidirectional case.  The difference
is that a processor has two ports (one sending and one receiving) rather than a single 
port that either sends or receives.  

The bidirectional greedy (bi-greedy) algorithm adheres to the following restrictions: the message is split using
equi-segmentation and the ports are filled giving priority to segments with smaller indices.  This restriction
differs from that of uni-greedy, in that a processor may receive a segment before sending a segment with a smaller
index.  For the uni-greedy algorithm, the smaller indexed segment must be sent first.
A downfall of this relaxation is that extra storage must be allocated.  However, other reduction algorithms also
require a larger buffer (see Rabenseifner~\cite{Rabenseifner:04}). 
  
\begin{algorithm*}[!p]

	\scriptsize{
		$S = \textmd{zeros}(p,1)$;  \\
		$R = \textmd{zeros}(p,1)$; \\
		$C = \textmd{zeros}(p,2)$; \\
		$t = 0$; \\
		$M = \textmd{zeros}(p,q)$;   \\
		\While{ $\min \{ M(i,j) \; | \; 1 \le i \le p \; \textmd{and} \; 1 \le j \le q \}  = 0$} {

			$\textmd{segStart} = \min \{ j \; | \; M(1,j) = 0 \}$; \\
			$\textmd{stop} = 1$; \\
			$j = \textmd{segStart} - 1$; \\
	
			\While{ $\textmd{stop}$ }{		
				$j \leftarrow j + 1$; \\
				$\textmd{COMP} = \{ i \; | \; C(i,1) \le t < C(i,2) \; \textmd{or} \; 
					t+ t_{comm}(j) > C(i)\}$; \\
				$I = \{ i \; | \;  M(i,j) = 0 \}$; \\  
				$\textmd{sendProc} = I \setminus \textmd{COMP} \setminus
					\{ i \; | \; S(i) > t \} $; \\
				$\textmd{recvProc} = I \setminus \textmd{COMP} \setminus
					\{ i \; | \; R(i) > t \} $; \\
	
				$\textmd{freeProc} = \textmd{sendProc} \bigcap \textmd{recvProc}$; \\
				$\textmd{sendProc} \leftarrow \textmd{sendProc} \setminus \textmd{freeProc}$; \\
				$\textmd{recvProc} \leftarrow \textmd{recvProc} \setminus \textmd{freeProc}$; \\
				
				$s = | \textmd{sendProc} |$; \\
				$r = | \textmd{recvProc} |$; \\
				$f = | \textmd{freeProc} |$; \\
				
				\uIf{ $s = r$ }{
					$y = \lfloor f/2 \rfloor$; \\
					$\textmd{sendProc} \leftarrow  \textmd{sendProc} \bigcup 
						\{ \textmd{freeProc}(i) \; | \; f-y+1 \le i \le f \}$; \\
					$\textmd{recvProc} \leftarrow  \textmd{recvProc} \bigcup 
						\{ \textmd{freeProc}(i) \; | \; 1 \le i \le y \}$; \\
				}
				\uElseIf{ $ s < r$}{
					$y = r - s$; \\
					$m = \min(f,y)$; \\ 
					$x = \lfloor (f - m)/2 \rfloor$; \\  
					\If{ $m >0$}{ $\textmd{sendProc} \leftarrow \textmd{sendProc} \bigcup
						\{ \textmd{freeProc}(i) \; | \; (f-(m+x) + 1 \le i \le f \}$; }
					\If{ $x >0$}{ $\textmd{recvProc} \leftarrow \textmd{recvProc} \bigcup
						\{ \textmd{freeProc}(i) \; | \; 1 \le i \le x \}$; }
				}
				\ElseIf{ $ r < s$}{
					$y = s - r$; \\
					$m = \min(f,y)$; \\ 
					$x = \lfloor (f - m)/2 \rfloor$; \\  
					\If{ $m >0$}{ $\textmd{recvProc} \leftarrow \textmd{recvProc} \bigcup
						\{ \textmd{freeProc}(i) \; | \; 1 \le i \le m+x \}$; }
					\If{ $x >0$}{ $\textmd{sendProc} \leftarrow \textmd{sendProc} \bigcup
						\{ \textmd{freeProc}(i) \; | \; f - x+1 \le i \le f \}$; }
				}
					
				$ l = \min( | \textmd{sendProc} | , | \textmd{recvProc} )$ \\
				\uIf{ $l = 0$ }{
					$\textmd{sendProc} \leftarrow \emptyset$; \\
					$\textmd{recvProc} \leftarrow \emptyset$; \\
				}
				\Else{
					$ \textmd{sendProc} = \{ \textmd{sendProc}(i) \; | \;  1 \le i \le l\}$; \\
					$ \textmd{recvProc} = \{ \textmd{recvProc}(i) \; | \; 1 \le i \le l\}$; \\
				}
				$M( i, j ) = t + t_{comm}(j), \forall \; i \; \textmd{s.t.} \; i \in \textmd{sendProc}$; \\
				$S( i ) = t + t_{comm}(j) , \forall \; i \; \textmd{s.t.} \; i \; \in \textmd{sendProc}$; \\
				$R( i ) = t + t_{comm}(j) + t_{comp}(j) , \forall \; i \; \textmd{s.t.} \; i \in \textmd{recvProc}$; \\
				$C(i,1) = t + t_{comm}(j), \forall \; i \; \textmd{s.t.} \; i \in \textmd{recvProc}$; \\
				$C(i,2) = t + t_{comm}(j) + t_{comp}(j) , \forall \; i \; \textmd{s.t.} \; i \in \textmd{recvProc}$; \\
			
				\If{ $| \{ i \; | \; M(i,j) = 0 \} | = 1 $ }{
					$ M(1,j) = max( \{ M(i,j) \; | \; 1 \le i \le p \} ) + t_{comp}(j) $; 
				}
	
				\uIf{ $ | \{ i \; | \; S(i) \le t \} | + | \{ i \; | \; R(i) \le t \} | < 2 $}{
					$\textmd{stop} = 0$;
				}
				\ElseIf{$j \ge q$}{ 
					$\textmd{stop} = 0$;
				}
			}
			$t = t + 1$; \\
		}
	}

\caption{Bi-greedy Algorithm}
\label{alg:bi-greedy}
\end{algorithm*}

Algorithm~\ref{alg:bi-greedy} on the next page gives the pseudo-code for the bi-greedy algorithm. $S$ and $R$ are
$p \times 1$ arrays where entry $i$ gives the time when 
processor $i-1$ has finished sending or receiving, respectively, the last scheduled communication.
$C$ is a $p \times 2$ array where $C(i,1)$ is the start time and $C(i,2)$ is the end time of the 
last scheduled computation on processor $i-1$.  $M$ is a $p \times q$ array where
entry $(i,j)$ gives the time when processor $i-1$ has finished sending segment $j$ (except for the root ($i=1$)
which gives the time after finishing the final computation). 
In each inner loop the maximum number of send/receive pairs for segment $j$ is determined. 
The final entries of sendProc and recvProc are
scheduled to send and receive segment $j$, respectively.  These processors are paired entry-wise to obtain
send/receive pairs.  The algorithm assumes a discrete time step.  

\begin{figure*}[!t]
	\centering
    \resizebox{.75\textwidth}{!}{\includegraphics{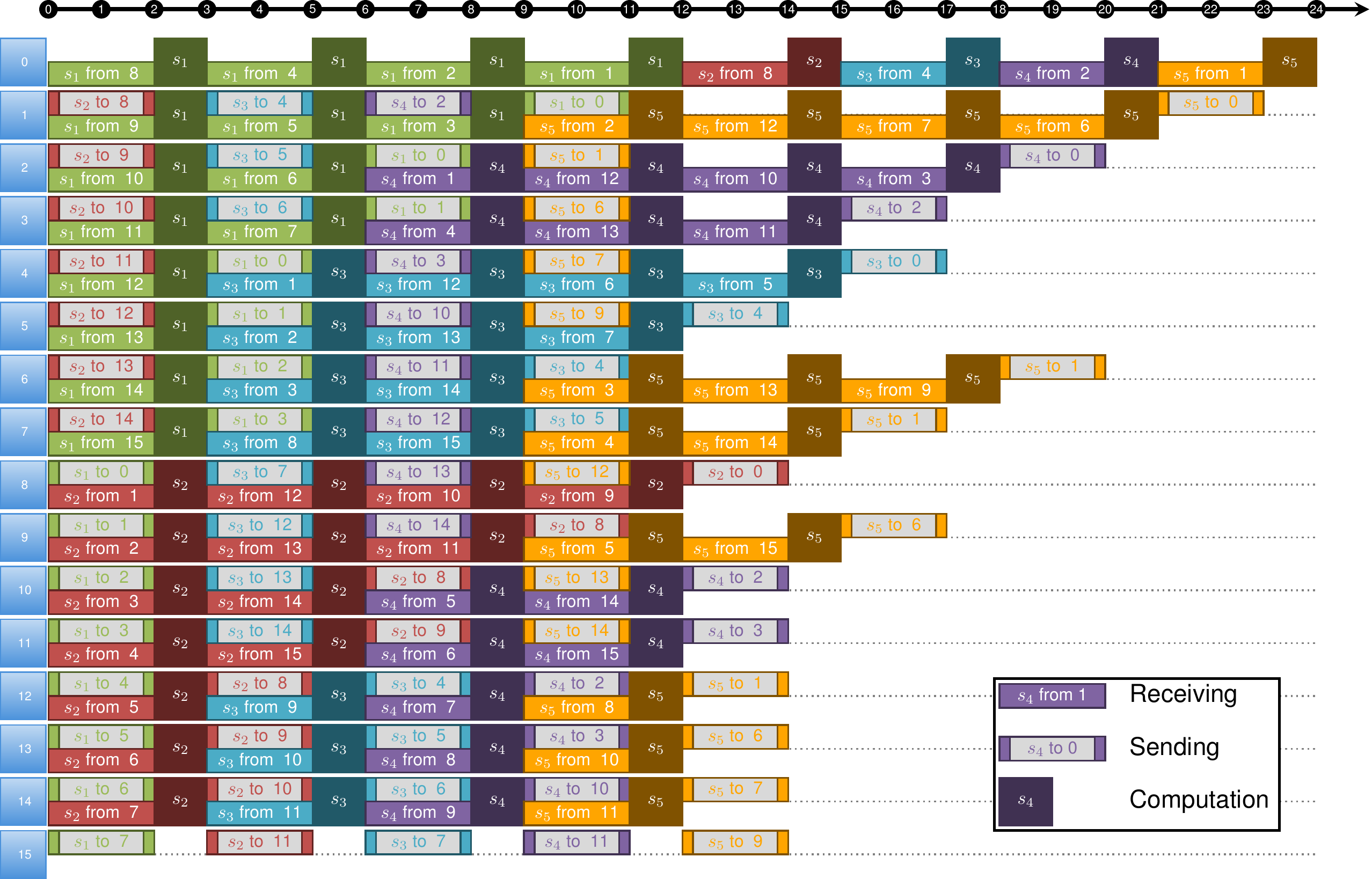}}
    \caption{\label{fig:bi-greedy-tree} Bi-greedy algorithm for 16 processors and 5
equi-segments for $t_{comm} = 2$ and $t_{comp} = 1$. Each color represents a different segment. 
Solid rectangles represent receiving processors and rectangles with end strips represent sending processors.  
The darker rectangles represent computation.}
\end{figure*}

Figure~\ref{fig:bi-greedy-tree} gives an example for 16 processors and 5 equi-segments with $t_{comm} = 2$ and
$t_{comp} = 1$. For this example we see that the time complexity matches the that of the optimal 
broadcast algorithm.  From experiments 
we conjecture that for all $p$ and $q$ the bi-greedy algorithm has the same time complexity as a reverse
optimal broadcast algorithm.
We do not provide a proof at this time.

\section{Numerical Results}\label{numer}
\begin{figure*}[htbp]
    \centerline{
        \subfloat[Message size versus ratio]{
            \label{fig:experimental_p64}
			\resizebox{.45\textwidth}{!}{\includegraphics{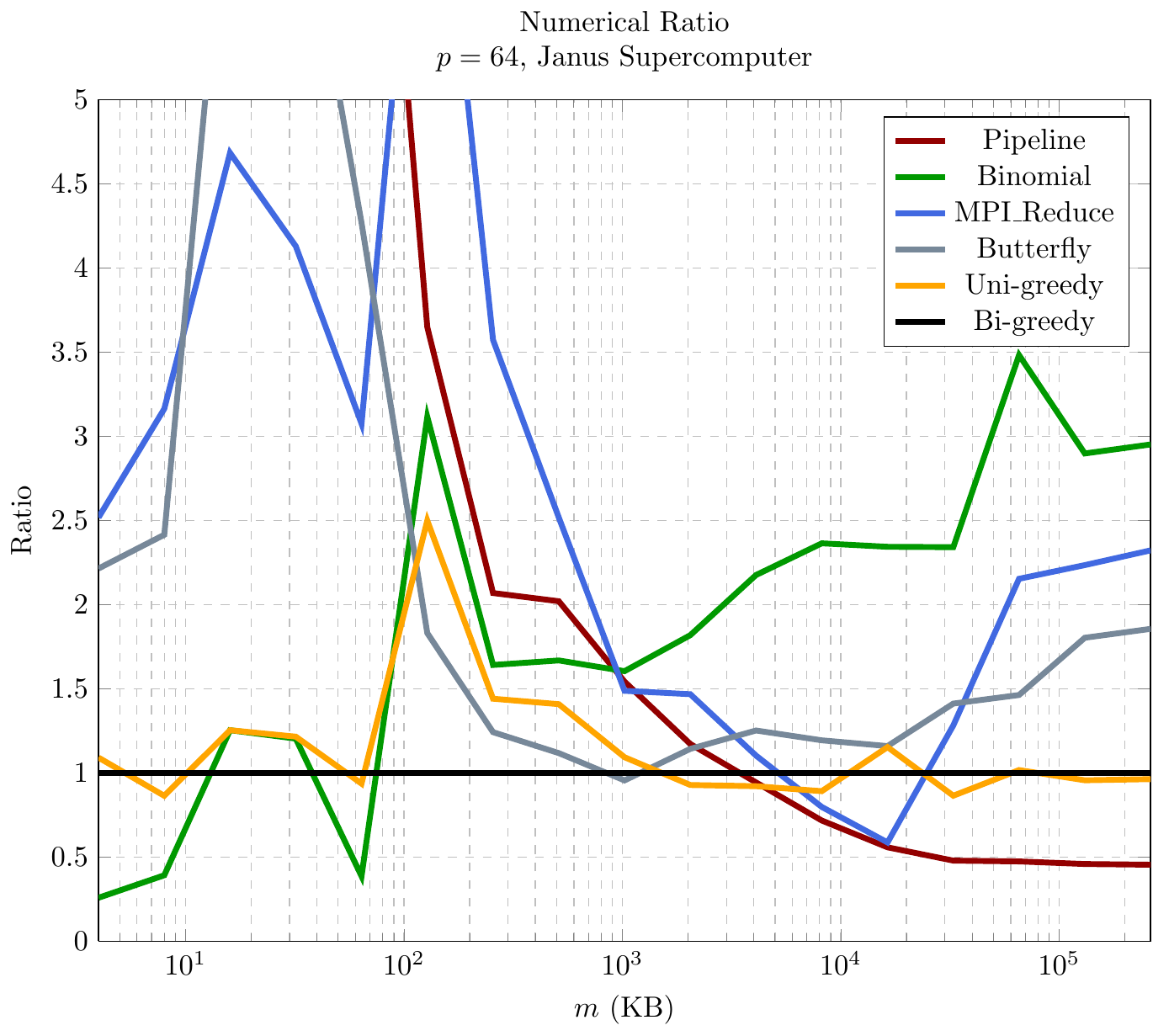}}
        }
        \hfil
        \subfloat[Message size versus time]{
            \label{fig:experimental_p64_time}
			\resizebox{.45\textwidth}{!}{\includegraphics{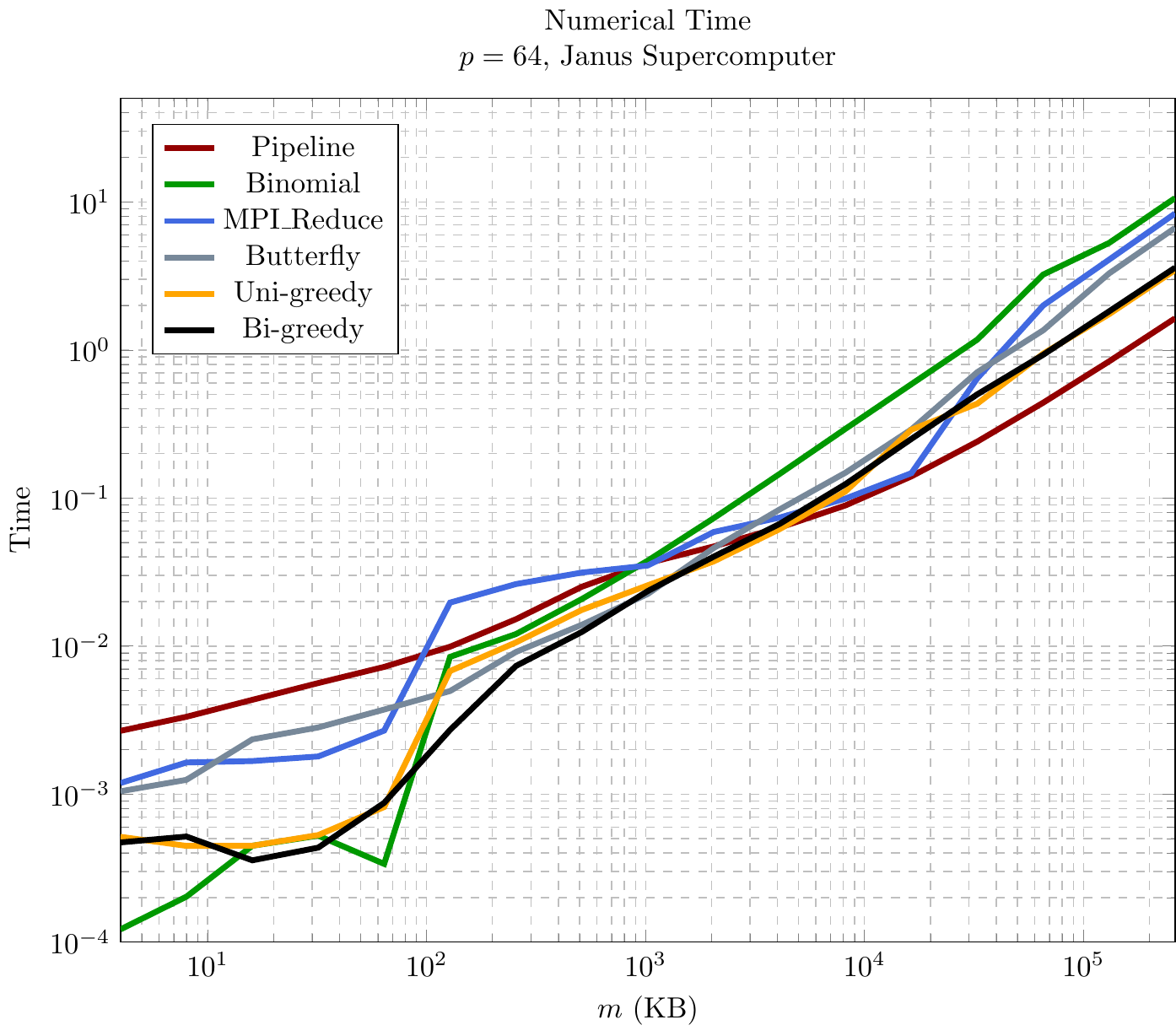}}
        }
    }
    \caption{Experimental results for each algorithm for all possible message
    sizes for 64 processors.}
    \label{fig:exp_p64}
\end{figure*}

For the numerical experiments we implemented the uni-greedy algorithm developed under the 
unidirectional contexted (Figure~\ref{fig:code_reduce_greedy}) and the bi-greedy algorithm developed under the
bidirectional context (Algorithm~\ref{alg:bi-greedy}). 
The algorithms were
implemented using OpenMPI version~1.4.3. All experiments were performed on the Janus supercomputer which consists 
of 1,368 nodes with each node having 12 cores.  Since processors on a single node have access to shared memory, 
only one processor per node was used with a total of 64 nodes in use for the experiments.  The 
nodes are connected using a fully non-blocking quad-data-rate (QDR) InfiniBand interconnect.
The time for a reduction was calculated by taking the minimum time of 10 experiments.

We wish to compare the algorithms assuming that each is optimally tuned for the best segmentation.
For a message of size $m$, segments of size $2^i,\;\textmd{s.t.}\;i = 0,\dots,\lfloor \log_2 m \rfloor$, 
where tested and the best time was taken to be the time of an optimally tuned algorithm.
Figures~\ref{fig:experimental_p64} and \ref{fig:experimental_p64_time} compares the optimally tuned greedy algorithms 
with other
algorithms for reduction. 

Comparing bi-greedy with binomial and pipeline shows similar results as to those of the
theoretical results. 
For small messages, the time for binomial and greedy are close, although binomial performs about twice as well
as greedy in a few cases.  For large messages,
pipeline begins to out perform all algorithms and is a factor of two better than bi-greedy. 
However, for medium size messages ($100 < m < 3000$ KB) bi-greedy is better than both binomial and pipeline.
The binary algorithm was not analyzed experimentally.
 
Comparing bi-greedy with the butterfly algorithm again shows similar results to the theoretical case.  
The bi-greedy algorithm outperforms butterfly for almost all message sizes. The butterfly
algorithm performs very well for medium size messages (about the same time as bi-greedy for $m=1000$ KB), but 
has poor performance for large and small messages.  
Comparing the bi-greedy algorithm with the built-in MPI function MPI\_Reduce, shows that greedy, 
in most cases, is the best. 
Finally, the uni-greedy algorithm (although designed for a unidirectional system) performed fairly well in general,
matching the performance of bi-greedy for small and large messages. 
 
For small messages, greedy is equivalent to binomial, but there is a fairly large difference in the
performance in this region.  This is because the two algorithms implemented use different send/receive pairings
and the performance differences indicate the system is actually heterogeneous.

Checking all of the possible segment sizes for a given message size is not practical.  
Auto-tuning could be used to determine the optimal segmentation for different message sizes.

\section{Unequal segmentation}\label{uneq}
So far we have only optimized the greedy algorithm by splitting a message into equally sized segments 
(except possibly the last segment).  An immediate question is then asked, what if the segments 
have unequal size?  A question we have not found to be considered in the literature.  Can the 
existing algorithms be improved? 

To investigate these questions the message size was fixed to $m =10$ and all possible 
segmentations were checked for $\alpha = 0,1,\dots,10,20,30,\dots,100,200,300,\dots,1000$, 
$\beta= 1$, $\gamma = 0,1$, and $p = 2^n \; \mbox{s.t.} \; n = 2,\dots,10$ and 
$p = 3(2^n) \; \mbox{s.t.} \; n=1,\dots,9$.  There a total of 512 possible segmentations.  
Examples of possible segmentations are $(9,1)$, $(3,3,3,1)$, $(3,2,4,1)$, $(2,2,1,5)$, 
$(6,3,1)$, etc.  An equi-segmentation is said to be one that 
has the same segment size for all segments except possibly the last segment may be smaller. 

For the greedy algorithm, 61 out of 986 experiments where optimized by an unequal segmentation.  
To compare unequal to equi-segmentation, we use the value 

\[ \textmd{ Ratio } = \dfrac{\textmd{ Best time for equi-segmentations}}{\textmd{Best time for all segmentations} }. \]

A ratio of 1 indicates that one of the optimal segmentations is an equi-segmentation.  
If there were multiple segmentations that were optimal and one of them was an equal 
segmentation, then the experiment was said to be optimized by an equi-segmentation.  
The maximum improvement over equi-segmentation was 7.3\%.  Of the 61 that 
did see and increase in performance the average improvement was 2.0\%.
 

\begin{table*}[p]
		\resizebox{\textwidth}{!}{ 
			
\begin{tabular}{c c}

\raisebox{-\height}{
\begin{tabular}{|c|c|c|c|} \hline
Parameters & Ratio & Equi-segmentation & Optimal Segmentation \\
\hline $p = 6, \alpha = 1, \gamma = 1$ & 1.0408 & (4,4,2) & (5,3,2) \\ 
& & &(3,5,2) \\ 
& & &(3,4,2,1) \\ 
& & &(4,2,2,2) \\ 
& & &(2,4,2,2) \\ 
& & &(3,2,2,2,1) \\ 
& & &(2,3,2,2,1) \\ 
\hline $p = 6, \alpha = 2, \gamma = 1$ & 1.0357 & (4,4,2) & (5,3,2) \\ 
& & &(3,5,2) \\ 
\hline $p = 8, \alpha = 0, \gamma = 1$ & 1.0571 & (1,1,1,1,1,1,1,1,1,1) & (2,1,1,1,1,1,1,1,1) \\ 
\hline $p = 8, \alpha = 1, \gamma = 1$ & 1.0200 & (4,4,2) & (5,2,2,1) \\ 
& & &(3,4,2,1) \\ 
\hline $p = 12, \alpha = 0, \gamma = 1$ & 1.0732 & (1,1,1,1,1,1,1,1,1,1) & (2,2,1,1,1,1,1,1) \\ 
\hline $p = 12, \alpha = 1, \gamma = 1$ & 1.0169 & (4,4,2) & (3,3,1,2,1) \\ 
& & &(2,2,1,2,2,1) \\ 
\hline $p = 12, \alpha = 3, \gamma = 1$ & 1.0130 & (5,5) & (4,5,1) \\ 
\hline $p = 16, \alpha = 1, \gamma = 0$ & 1.0526 & (3,3,3,1) & (5,3,2) \\ 
& & &(4,2,2,2) \\ 
\hline $p = 16, \alpha = 1, \gamma = 1$ & 1.0161 & (3,3,3,1) & (3,2,2,2,1) \\ 
\hline $p = 16, \alpha = 2, \gamma = 1$ & 1.0278 & (4,4,2) & (4,3,3) \\ 
\hline $p = 24, \alpha = 3, \gamma = 1$ & 1.0108 & (3,3,3,1) & (5,4,1) \\ 
& & &(4,3,3) \\ 
\hline $p = 24, \alpha = 4, \gamma = 1$ & 1.0388 & (5,5) & (5,4,1) \\ 
& & &(4,3,3) \\ 
\hline $p = 32, \alpha = 1, \gamma = 0$ & 1.0222 & (4,4,2) & (3,3,2,2) \\ 
\hline $p = 32, \alpha = 1, \gamma = 1$ & 1.0141 & (2,2,2,2,2) & (3,2,2,2,1) \\ 
\hline $p = 32, \alpha = 2, \gamma = 1$ & 1.0118 & (4,4,2) & (4,3,3) \\ 
& & &(5,2,2,1) \\ 
& & &(3,3,2,2) \\ 
\hline $p = 32, \alpha = 3, \gamma = 1$ & 1.0104 & (4,4,2) & (4,3,3) \\ 
\hline $p = 48, \alpha = 1, \gamma = 0$ & 1.0408 & (4,4,2) & (4,2,2,2) \\ 
& & &(3,2,3,2) \\ 
\hline $p = 48, \alpha = 1, \gamma = 1$ & 1.0260 & (2,2,2,2,2) & (3,2,2,2,1) \\ 
& & &(2,2,1,2,2,1) \\ 
\hline $p = 48, \alpha = 2, \gamma = 0$ & 1.0164 & (4,4,2) & (5,3,2) \\ 
\hline $p = 48, \alpha = 2, \gamma = 1$ & 1.0215 & (3,3,3,1) & (3,3,2,2) \\ 
\hline $p = 48, \alpha = 3, \gamma = 1$ & 1.0093 & (4,4,2) & (4,3,3) \\ 
& & &(3,4,3) \\ 
& & &(3,3,2,2) \\ 
& & &(3,2,3,2) \\ 
\hline $p = 48, \alpha = 4, \gamma = 1$ & 1.0084 & (4,4,2) & (5,4,1) \\ 
& & &(5,2,3) \\ 
& & &(4,3,3) \\ 
& & &(3,4,3) \\ 
\hline $p = 64, \alpha = 1, \gamma = 1$ & 1.0253 & (3,3,3,1) & (3,2,2,2,1) \\ 
& & &(2,2,2,2,1,1) \\ 
& & &(2,2,2,1,2,1) \\ 
& & &(2,2,2,1,1,1,1) \\ 
& & &(2,2,1,1,1,1,1,1) \\ 
\hline $p = 64, \alpha = 2, \gamma = 1$ & 1.0106 & (3,3,3,1) & (3,3,2,2) \\ 
\hline $p = 64, \alpha = 3, \gamma = 1$ & 1.0093 & (4,4,2) & (3,3,2,2) \\ 
\hline $p = 96, \alpha = 1, \gamma = 0$ & 1.0182 & (4,4,2) & (3,3,2,2) \\ 
& & &(3,2,3,2) \\ 
& & &(2,3,3,2) \\ 
\hline $p = 96, \alpha = 1, \gamma = 1$ & 1.0119 & (1,1,1,1,1,1,1,1,1,1) & (2,2,2,2,1,1) \\ 
& & &(2,2,2,1,1,1,1) \\ 
& & &(2,2,1,1,1,1,1,1) \\ 
\hline $p = 96, \alpha = 2, \gamma = 1$ & 1.0098 & (3,3,3,1) & (3,3,2,2) \\ 
\hline $p = 96, \alpha = 3, \gamma = 1$ & 1.0085 & (3,3,3,1) & (3,3,2,2) \\ 
\hline $p = 128, \alpha = 1, \gamma = 1$ & 1.0116 & (2,2,2,2,2) & (2,2,2,1,1,1,1) \\ 
& & &(2,1,1,1,1,1,1,1,1) \\ 
& & &(1,1,1,2,1,1,1,1,1) \\ 
\hline $p = 192, \alpha = 1, \gamma = 0$ & 1.0169 & (3,3,3,1) & (2,2,2,2,1,1) \\ 
\hline
\end{tabular}
}
&
\raisebox{-\height}{
\begin{tabular}{|c|c|c|c|} \hline
Parameters & Ratio & Equi-segmentation & Optimal Segmentation \\
\hline $p = 256, \alpha = 1, \gamma = 0$ & 1.0161 & (2,2,2,2,2) & (4,3,2,1) \\ 
& & &(3,2,3,2) \\ 
& & &(2,3,3,2) \\ 
& & &(3,3,2,1,1) \\ 
& & &(2,2,3,2,1) \\ 
& & &(2,3,2,1,1,1) \\ 
& & &(3,2,1,2,1,1) \\ 
& & &(3,1,2,2,1,1) \\ 
& & &(2,2,2,2,1,1) \\ 
& & &(3,2,1,1,2,1) \\ 
& & &(3,1,2,1,2,1) \\ 
& & &(2,2,2,1,2,1) \\ 
& & &(2,2,1,2,2,1) \\ 
& & &(2,1,2,2,2,1) \\ 
& & &(1,2,2,2,2,1) \\ 
& & &(2,1,2,2,1,2) \\ 
\hline $p = 256, \alpha = 1, \gamma = 1$ & 1.0109 & (1,1,1,1,1,1,1,1,1,1) & (2,1,2,1,1,1,1,1) \\ 
& & &(2,1,1,1,1,1,1,1,1) \\ 
\hline $p = 256, \alpha = 2, \gamma = 0$ & 1.0256 & (4,4,2) & (5,3,2) \\ 
& & &(4,3,3) \\ 
\hline $p = 256, \alpha = 2, \gamma = 1$ & 1.0085 & (2,2,2,2,2) & (3,2,3,2) \\ 
& & &(2,2,2,3,1) \\ 
& & &(2,1,2,2,2,1) \\ 
& & &(2,1,2,2,1,1,1) \\ 
\hline $p = 256, \alpha = 3, \gamma = 0$ & 1.0217 & (4,4,2) & (4,3,3) \\ 
\hline $p = 256, \alpha = 3, \gamma = 1$ & 1.0224 & (3,3,3,1) & (3,2,3,2) \\ 
\hline $p = 256, \alpha = 4, \gamma = 1$ & 1.0265 & (3,3,3,1) & (5,3,2) \\ 
& & &(4,3,3) \\ 
\hline $p = 256, \alpha = 5, \gamma = 1$ & 1.0303 & (4,4,2) & (5,3,2) \\ 
\hline $p = 256, \alpha = 6, \gamma = 1$ & 1.0279 & (4,4,2) & (5,3,2) \\ 
\hline $p = 384, \alpha = 3, \gamma = 1$ & 1.0142 & (3,3,3,1) & (3,2,3,2) \\ 
& & &(2,3,3,2) \\ 
\hline $p = 512, \alpha = 1, \gamma = 0$ & 1.0154 & (2,2,2,2,2) & (2,2,2,2,1,1) \\ 
& & &(2,2,2,1,2,1) \\ 
\hline $p = 512, \alpha = 2, \gamma = 0$ & 1.0238 & (4,4,2) & (4,3,3) \\ 
\hline $p = 512, \alpha = 3, \gamma = 0$ & 1.0100 & (4,4,2) & (4,3,3) \\ 
\hline $p = 512, \alpha = 3, \gamma = 1$ & 1.0069 & (3,3,3,1) & (3,2,3,2) \\ 
& & &(2,3,3,2) \\ 
\hline $p = 512, \alpha = 4, \gamma = 1$ & 1.0061 & (3,3,3,1) & (2,3,3,2) \\ 
\hline $p = 512, \alpha = 5, \gamma = 1$ & 1.0167 & (4,4,2) & (4,3,3) \\ 
\hline $p = 512, \alpha = 6, \gamma = 1$ & 1.0154 & (4,4,2) & (4,3,3) \\ 
\hline $p = 512, \alpha = 7, \gamma = 1$ & 1.0095 & (4,4,2) & (4,3,3) \\ 
\hline $p = 768, \alpha = 1, \gamma = 0$ & 1.0147 & (2,2,2,2,2) & (2,2,2,2,1,1) \\ 
& & &(2,2,2,1,2,1) \\ 
& & &(2,1,2,2,2,1) \\ 
& & &(1,2,2,2,2,1) \\ 
\hline $p = 768, \alpha = 1, \gamma = 1$ & 1.0099 & (1,1,1,1,1,1,1,1,1,1) & (2,1,1,1,1,1,1,1,1) \\ 
\hline $p = 768, \alpha = 2, \gamma = 0$ & 1.0111 & (3,3,3,1) & (4,3,3) \\ 
& & &(3,4,2,1) \\ 
& & &(4,2,3,1) \\ 
& & &(3,2,3,2) \\ 
& & &(2,3,3,2) \\ 
& & &(2,2,4,2) \\ 
& & &(2,2,3,3) \\ 
\hline $p = 768, \alpha = 3, \gamma = 0$ & 1.0189 & (4,4,2) & (4,3,3) \\ 
\hline $p = 768, \alpha = 3, \gamma = 1$ & 1.0066 & (2,2,2,2,2) & (2,3,3,2) \\ 
\hline $p = 768, \alpha = 4, \gamma = 0$ & 1.0164 & (4,4,2) & (4,3,3) \\ 
\hline $p = 768, \alpha = 4, \gamma = 1$ & 1.0174 & (3,3,3,1) & (3,2,3,2) \\ 
& & &(2,3,3,2) \\ 
\hline $p = 768, \alpha = 5, \gamma = 1$ & 1.0317 & (3,3,3,1) & (4,3,3) \\ 
\hline $p = 768, \alpha = 6, \gamma = 1$ & 1.0341 & (4,4,2) & (4,3,3) \\ 
\hline $p = 768, \alpha = 7, \gamma = 1$ & 1.0270 & (4,4,2) & (4,3,3) \\ 
\hline $p = 768, \alpha = 8, \gamma = 1$ & 1.0209 & (4,4,2) & (4,3,3) \\ 
\hline $p = 1024, \alpha = 3, \gamma = 0$ & 1.0093 & (4,4,2) & (4,3,3) \\ 

\hline
\end{tabular}
}
\end{tabular}

		}
    \caption{Optimal Segmentation compared to best of the equi-segmentations.}
    \label{Opt_seg_all}
\end{table*}

Table~\ref{Opt_seg_all} shows results for all experiments that were 
optimized by unequal segmentation. We note that, in all cases, one of the optimal 
segmentations was nearly the same as the best equi-segmentation. Actually, to 
obtain the optimal (unequal) segmentation from the best of the equi-segmentations, 
the size of only one or two segments had to be increased or decreased by a value of 
1 in most cases.

\begin{figure*}[!tb]
    \centerline{
        \subfloat[$\gamma = 0$]{
            \includegraphics[width=0.45\textwidth]{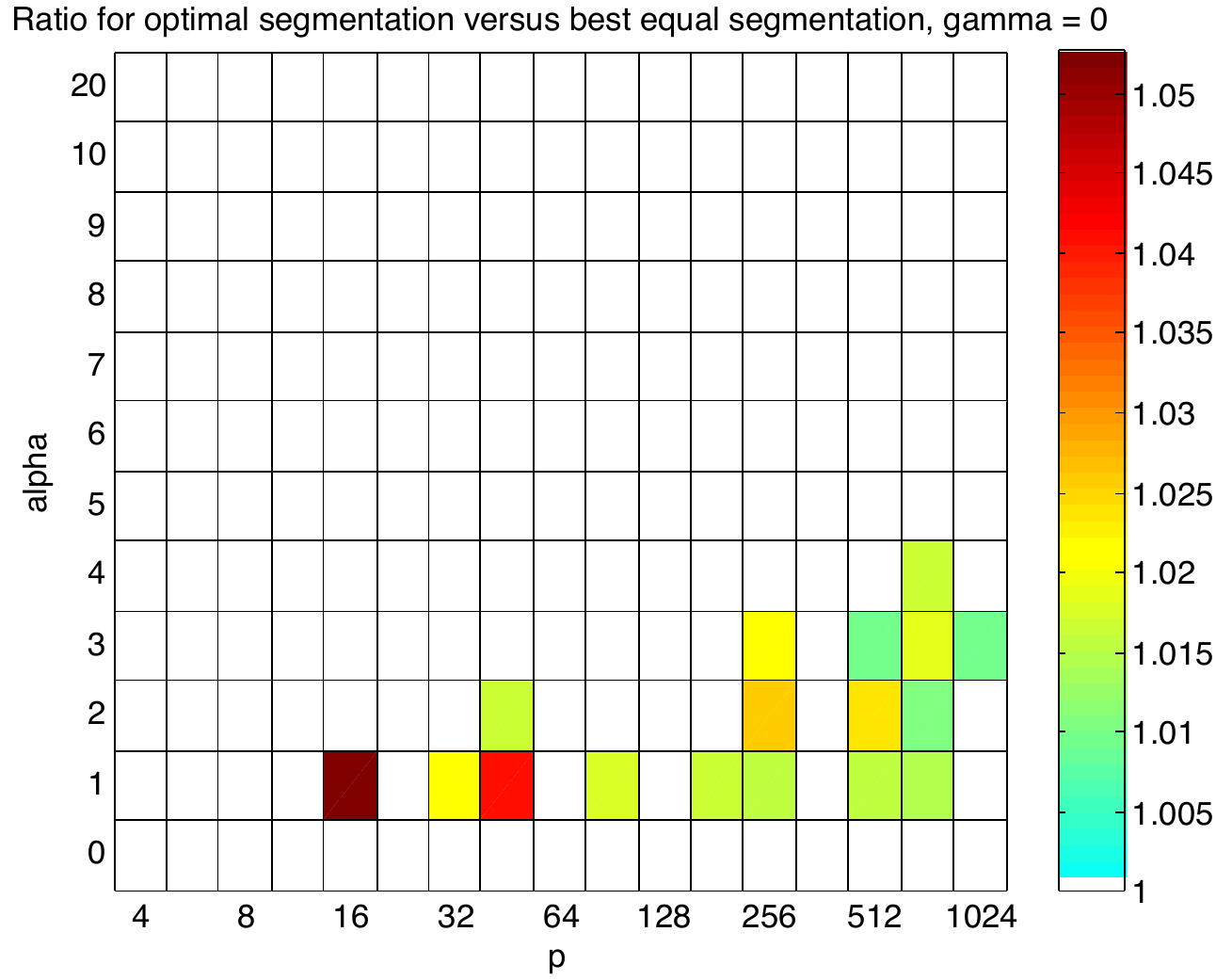}%
            \label{fig:unequal0}
        }
        \hfil
        \subfloat[$\gamma = 1$]{
            \includegraphics[width=0.45\textwidth]{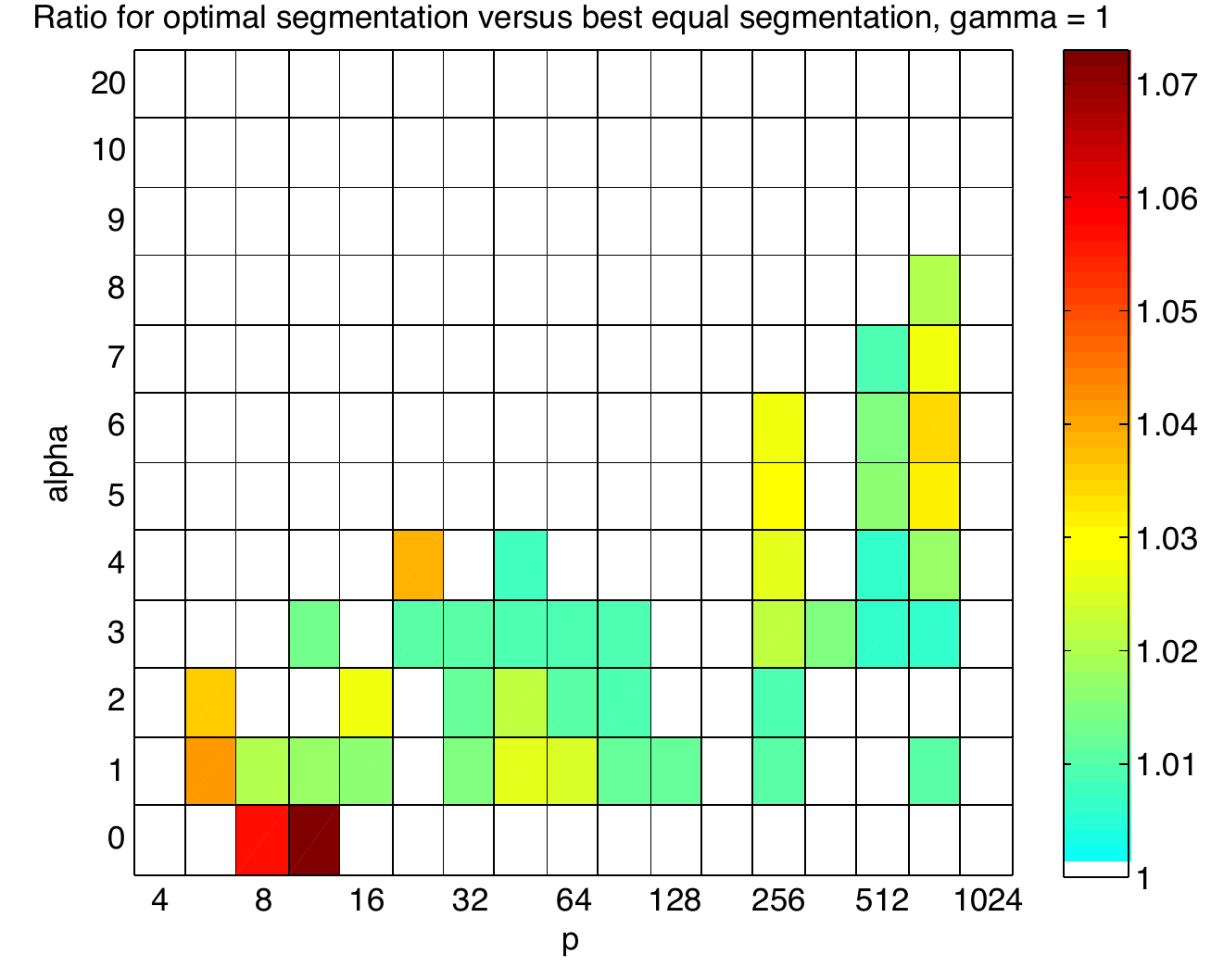}%
            \label{fig:unequal1}
        }
    }
    \caption{Ratio for optimal segmentation versus the best equal segmentation.}
    \label{fig:unequal}
\end{figure*}

Figures~\ref{fig:unequal0} and \ref{fig:unequal1} show the results for all $\alpha$, $\gamma$, and $p$.  
For $\alpha > 20$, the greedy algorithm was always optimized by an equi-segmentation and therefore is 
not graphed.   A pattern as to when unequal segmentation is optimal is not evident from the figures.  
For the pipeline algorithm, all experiments where optimized by equi-segmentation.  The binary 
algorithm was not checked for unequal segmentation.

The experiment is limited to small values of $m$ since the number of possible segmentations 
grows exponentially as $m$ increases.  The algorithms were not implemented for unequal 
segmentations since the small theoretical improvements most likely will not be obtained.  
Further investigation is required to determine if for larger messages sizes does the 
improvement of unequal segmentation become greater.

\section{Conclusion}\label{conc}
Two new algorithms for reduction are presented.  The two algorithms are developed 
under two different communication models (unidirectional and bidirectional).  

We prove that the unidirectional algorithm is optimal under certain assumptions.
Our assumptions are fully connected, homogeneous system, and processing
the segments in order.  
Previous algorithms that satisfy the same assumptions are only appropriate for some configurations. 
The uni-greedy algorithm is optimal for all configurations. 
The most improvement over standard algorithms is for messages of ``medium''
length, providing about 50\% improvement when compared to the best of
the standard algorithms in that region.  The region of ``medium'' length
messages become more prevalent as the number of processors increases.

We adapted the greedy algorithm that was developed in the unidirectional context to an algorithm
that was more suited for a bidirectional system.
With simulations we found that the bi-greedy algorithm matched the
time complexity of optimal broadcast algorithms (scheduled in reverse order as a reduction).
Similar performance improvements where found as in the unidirectional case.

Implementation of the algorithms confirm the theoretical results.  Our implementation of the
bi-greedy algorithm is was best among all algorithms implemented for ``medium'' size messages. 
For small and large messages, the more simplistic algorithms (binomial and pipeline), which are
asymptotically optimal, outperformed the greedy algorithms.  

Finally, the concept of unequal segmentation was discussed and analyzed for the greedy
algorithm and the pipeline algorithm in a unidirectional system.  For the greedy algorithm, unequal
segmentation provided, in our sample test suite, the optimal segmentation $6.2\%$ of the time with a
maximum improvement of 7.3\%. The pipeline algorithm was always optimized
by equal segmentation.

\section*{Acknowledgment}
This work utilized the Janus supercomputer, which is supported by the National
Science Foundation (award number CNS-0821794) and the University of Colorado
Boulder. The Janus supercomputer is a joint effort of the University of
Colorado Boulder, the University of Colorado Denver and the National Center for
Atmospheric Research.

\bibliographystyle{IEEEtran}
\bibliography{biblio}
\end{document}